\PassOptionsToPackage{dvipdfmx}{graphicx}
\PassOptionsToPackage{dvipdfmx}{xcolor}
\PassOptionsToPackage{dvipdfmx}{hyperref}

\documentclass[3p,nonatbib]{elsarticle}

\usepackage{amsmath,amssymb,amsthm}
\usepackage{graphicx,xcolor,caption}
\usepackage[normalem]{ulem}
\usepackage{arydshln}
\usepackage[hidelinks,hypertexnames=false]{hyperref}

\newtheorem{theorem}{Theorem}
\newtheorem{lemma}{Lemma}
\newtheorem{proposition}{Proposition}
\theoremstyle{definition}
\newtheorem{example}{Example}
\theoremstyle{remark}
\newtheorem{remark}{Remark}

\journal{Physica D: Nonlinear Phenomena}

\begin{document}

\begin{frontmatter}

\title{An integrable semi-discretization of the two-component Hunter--Saxton equation}

\author[was-pam]{Ayako Hori\corref{cor1}}
\ead{ayako0903@akane.waseda.jp}
\author[was-pam]{Yuta Tanaka}
\author[was-am]{Ken-ichi Maruno}
\author[kobe]{Yasuhiro Ohta}

\cortext[cor1]{Corresponding author}

\address[was-pam]{Department of Pure and Applied Mathematics, School of Fundamental Science and Engineering, Waseda University, 3-4-1 Okubo, Shinjuku-ku, Tokyo 169-8555, Japan}
\address[was-am]{Department of Applied Mathematics, Faculty of Science and Engineering, Waseda University, 3-4-1 Okubo, Shinjuku-ku, Tokyo 169-8555, Japan}
\address[kobe]{Department of Mathematics, Kobe University, 1-1 Rokkodai-cho, Nada-ku, Kobe 657-8501, Japan}

\begin{abstract}
In this paper, we propose an integrable semi-discretization of 
the two-component Hunter--Saxton (2-HS) equation, the short-wave limit of the two-component Camassa--Holm (2-CH) equation. 
At the continuous level, we show that the 2-HS equation can be derived from a new bilinear formulation, distinct from the conventional one in the literature, via a pseudo 2-reduction and 
a hodograph transformation. For the semi-discrete construction, we first discretize the underlying bilinear equations in the spatial direction. We then impose the pseudo 2-reduction and apply a discrete hodograph transformation to obtain the semi-discrete system in the physical variables. To the best of our knowledge, the resulting system is the first integrable semi-discretization that preserves 
the two-component structure of the 2-HS equation. We construct the N-soliton solutions of the continuous and 
semi-discrete systems in Wronskian and Casoratian forms, respectively. The integrability of the semi-discrete system is inherited from the underlying integrable hierarchy and is further verified by a Lax pair.
\end{abstract}

\begin{keyword}
integrable semi-discretization \sep two-component Hunter--Saxton equation \sep
short-wave limit of the two-component Camassa--Holm equation \sep hodograph transformation
\end{keyword}

\end{frontmatter}

\section{Introduction}
\label{sec_intro}
The study of discrete integrable systems has progressed
significantly since the mid-1970s, 
following the pioneering work of Hirota and of Ablowitz and Ladik. 
Hirota developed the bilinear and $\tau$-function approach, 
which enables soliton equations to be discretized while preserving their solution structures, 
and successfully applied it to equations 
such as the KdV and sine-Gordon equations \cite{Hirota1, Hirota2, Hirota3, Hirota4, Hirota5}. 
In particular, Hirota proposed the discrete analogue of a generalized Toda equation (DAGTE) \cite{Hirota_DAGTE}, 
which is now known as the discrete Kadomtsev-Petviashvili (KP) equation. 
Subsequently, within the Sato-theoretic framework, the algebraic structures of these discrete equations 
were extensively studied by Date, Jimbo, and Miwa \cite{Date1, Date2, Date3, Date4, Date5}.
In this context, Miwa introduced a key change of variables, now known as the Miwa transformation, 
which provides a direct connection between continuous and discrete integrable systems, 
and showed that the discrete KP equation, also referred to as the Hirota-Miwa equation, 
serves as a master equation for integrable hierarchies \cite{HM}. 
Independently, Ablowitz and Ladik introduced a method of integrable discretization based on the AKNS-type linear 
eigenvalue problem and successfully constructed integrable discretizations of equations 
such as the nonlinear Schr\"{o}dinger equation \cite{Abl1, Abl2, Abl3, Abl4, Abl5}.

As the study of discrete integrable systems progressed, several additional approaches 
to integrable discretization were proposed. For example, in the 1980s, the direct linearization method 
was introduced by Nijhoff, Capel, Quispel, et al. \cite{DIM1, DIM2}. 
Furthermore, Suris developed a systematic approach to integrable discretization 
based on Lax matrices and Hamiltonian structures using Miura-type transformations (localizing changes of variables) to construct local discrete equations of motion from non-local ones \cite{Suris}.
For a comprehensive historical account of discrete integrable systems, see Suris \cite{Suris} and Hietarinta-Joshi-Nijhoff \cite{DS}. 

It is well known that there exist soliton equations associated with 
Wadati-Konno-Ichikawa (WKI)-type linear eigenvalue problems \cite{WKI1, WKI2}. 
Ishimori, as well as Wadati and Sogo, showed that soliton equations in the WKI class can be transformed 
into soliton equations associated with 
AKNS-type linear eigenvalue problems through a hodograph transformation, 
also known as a reciprocal transformation \cite{ishimori1,ishimori2,wadachi1, rogers}. 
Many WKI-type soliton equations possess singular solutions, and their integrable discretizations have long been considered difficult. 
However, by incorporating a discrete hodograph transformation, some of the present authors have constructed integrable discretizations 
of soliton equations with singularities \cite{maruno1,maruno2,maruno3,maruno4,maruno5,maruno6,maruno7}. 

A prominent example of such equations is the Camassa--Holm (CH) equation. The CH equation
 \begin{eqnarray}
 u_{t}+2\kappa^{2}u_{x}-u_{txx}+3uu_{x}=2u_{x}u_{xx}+uu_{xxx},\label{CH}
\end{eqnarray}
where $\kappa$ is a constant, was derived as a model equation for shallow-water waves \cite{CH1, CH2}. This equation is known 
to admit peakon solutions when $\kappa=0$, while for $\kappa\neq0$ it possesses cuspon and smooth soliton solutions. 
Because of the singularities of cuspon and peakon solutions, an integrable discretization of the CH equation had not been constructed. 
Subsequently, some of the present authors succeeded in discretizing the CH equation 
in a spatial direction by reducing it to the bilinear form 
of the extended 2D Toda lattice hierarchy through a hodograph transformation and 
applying Hirota's bilinear method with the discrete hodograph transformation \cite{maruno1}. 
The discretization automatically refines the mesh in regions where the solution varies rapidly, so that the resulting integrable semi-discrete equation serves as a self-adaptive moving-mesh scheme.
The same authors also confirmed its effectiveness as a numerical scheme \cite{maruno12}.

Taking the short-wave limit of the CH equation (\ref{CH}) reduces it to the Hunter--Saxton (HS) equation:
\begin{eqnarray}
u_{txx}-2\kappa^{2}u_{x}+2u_{x}u_{xx}+uu_{xxx}=0.\label{HS}
\end{eqnarray}
In the case of $\kappa=0$, the HS equation (\ref{HS}) reduces to the classical HS equation, 
which is a model for weakly nonlinear orientation waves in massive nematic liquid crystals \cite{HS-1}. 
Hunter and Zheng investigated the bi-Hamiltonian structure and the Lax pair, 
thereby demonstrating the integrability of the classical HS equation \cite{HS-2}. 
They also studied the global existence of weak solutions and related properties in their subsequent works \cite{HS-3, HS-4}. 

In the present paper, we study the two-component HS (2-HS) equation
\begin{eqnarray}
\left\{
\begin{array}{ll}
 u_{txx}-2\kappa^{2} u_{x}-2u_{x}u_{xx}-uu_{xxx}=-\mu\rho\rho_{x},\\
\rho_{t}-(\rho u)_{x}=0,
\end{array}
\right.
\label{Feng}
\end{eqnarray}
which is a two-component generalization of the HS equation (\ref{HS}). 
The parameter $\kappa^{2}$ is a non-negative real number, and $\mu$ is a non-zero real number. 
Historically, a system equivalent to the
$\kappa=0$ case of (\ref{Feng}) already appeared as the scaling limit of the
dual Ito system in the work of Olver and Rosenau \cite{2CH-0}, although it was
not identified there by the name 2-HS. Pavlov later formulated the system explicitly as the two-component Hunter--Saxton system, related it to the
Gurevich--Zybin system and the first negative flow of the two-component
Harry--Dym hierarchy, and gave a scalar spectral representation
\cite{Pavlov-GZ}. Aratyn, Gomes, and Zimerman subsequently embedded the
two-component CH and Dym/HS-type systems in a unified deformation framework
and retained the integration constant in the corresponding scalar spectral
problem \cite{AGZ}. The nonzero linear term in (\ref{Feng}) is associated
with this integration constant.
Subsequent analytical and geometric studies include work by
Wunsch and collaborators on weak and global solutions in periodic settings
with $\kappa=0$ \cite{2HS-1,2HS-3,2HS-4,2HS-6}, by Moon and Liu on wave
breaking and global existence \cite{2HS-7}, and by Kohlmann on
semidirect-product geometry \cite{2HS-5}.
Related studies include supersymmetric extensions of the CH and HS equations \cite{2HS-2} and solitary-wave solutions of a generalized two-component Hunter--Saxton system \cite{2HS-8}.

The 2-HS equation (\ref{Feng}) can be obtained as the short-wave limit of the
2-CH equation \cite{2CH-0,AGZ,2CH-1,2CH-2,2CH-3}
\begin{eqnarray}
\left\{
\begin{array}{ll}
 m_{t}-um_{x}-2mu_{x}+\mu \rho \rho_{x}=0,\\
 \rho_{t}-(\rho u)_{x}=0,\\
 m=\kappa^{2}-u+u_{xx},
\end{array}
\right.
\label{2CH}
\end{eqnarray}
which is a two-component generalization of the CH equation (\ref{CH}). The parameter $\kappa^{2}$ is a non-negative real number, 
and $\mu$ is a non-zero real number. A system equivalent to
(\ref{2CH}) arose from the
tri-Hamiltonian dual of the Ito system constructed by Olver and Rosenau
\cite{2CH-0}. Related bi-Hamiltonian and deformation formulations, including
the $\mu=-1$ case, were subsequently developed in
\cite{AGZ,2CH-1,2CH-2,2CH-3}.
From a physical perspective, it has been shown that the 2-CH equation (\ref{2CH}) with $\mu=1$ and $\kappa=0$ can be derived from the fully nonlinear Green-Naghdi equations, a shallow water model \cite{2CH-00}.
In the case of $\rho=0$, the 2-CH equation (\ref{2CH}) reduces, up to the sign convention for $u$, to the CH equation (\ref{CH}). 

Now, we take the short-wave limit of the 2-CH equation (\ref{2CH}) to derive the 2-HS equation (\ref{Feng}). 
Following the procedure in \cite{maruno2}, we introduce the time and space variables $\tau$ and $\xi$
 \begin{eqnarray}
\tau=\epsilon t,\qquad \xi=\epsilon^{-1} x, \label{SL}
\end{eqnarray}
where $\epsilon$ is a small positive parameter.  We expand $u$ and $\rho$ in powers of $\epsilon$ as 
\begin{equation}
u=\epsilon^{2}(u_{0}+\epsilon u_{1}+\epsilon^{2} u_{2}+\cdots),\quad \rho=\epsilon(\rho_{0}+\epsilon \rho_{1}+\epsilon^{2} \rho_{2}+\cdots),\label{1.6}
\end{equation}
where $u_{i}$ and $\rho_{i}$ $(i=0,1,2,\cdots)$ are functions of $\xi$ and $\tau$. From (\ref{SL}), we obtain
\begin{equation} 
\partial_{t}=\epsilon\partial_{\tau},\quad\partial_{x}=\epsilon^{-1}\partial_{\xi}.\label{1.7}
\end{equation}
Substituting (\ref{1.6}) and (\ref{1.7}) into the first and second equations of (\ref{2CH}), we obtain 
the following PDEs for $u_{0}$ and $\rho_{0}$ at the lowest order in $\epsilon$:
\begin{equation}
u_{0,\tau\xi\xi}-2\kappa^{2} u_{0,\xi}-2u_{0,\xi}u_{0,\xi\xi}-u_{0}u_{0,\xi\xi\xi}=-\mu \rho_{0}\rho_{0,\xi},\label{1.8}
\end{equation}
\begin{equation}
\rho_{0,\tau}-(u_{0}\rho_{0})_{\xi}=0.\label{1.9}
\end{equation}
Rewriting the variables as $u_{0} \rightarrow u$, $\rho_{0} \rightarrow \rho$, $\xi \rightarrow x$, 
and $\tau \rightarrow t$ in (\ref{1.8}) and (\ref{1.9}), we obtain the 2-HS equation (\ref{Feng}).

For $\kappa\neq0$, the coefficient $2\kappa^{2}$ in (\ref{Feng}) can be normalized to $4$ by scaling $x$, $t$, and $\rho$. 
For the case $\mu$<0, writing $\mu=-c^{2}$ with $c>0$, we use the normalized form
\begin{eqnarray}
\left\{
\begin{array}{ll}
u_{txx}-4 u_{x}-2u_{x}u_{xx}-uu_{xxx}=c^{2}\rho\rho_{x},\\
\rho_{t}=(\rho u)_{x},
 \end{array}
\right.
\label{2HHS}
\end{eqnarray}
where $c$ is a positive constant. 
Integrating (\ref{2HHS}) with respect to $x$ under the boundary conditions $u(x,t)\rightarrow 0$ and $\rho(x,t)\rightarrow 1$ 
as $|x|\rightarrow \infty$ yields
\begin{eqnarray}
\left\{
\begin{array}{ll}
u_{tx}-uu_{xx}-\displaystyle\frac{1}{2}u_{x}^{2}-4u-\displaystyle\frac{c^{2}}{2}\rho^{2}+\displaystyle\frac{c^{2}}{2}=0,\\
\rho_{t}=(\rho u)_{x}.
 \end{array}
\right.
\label{2HS}
\end{eqnarray}
Hereafter, we refer to (\ref{2HS}) as the 2-HS equation.
For the normalized 2-HS equation (\ref{2HHS}), particularly in the field of integrable systems, a set of bilinear equations was derived through a 2-reduction from the extended two-dimensional Toda lattice equation, 
which coincides with the bilinear form of the CH equation, and the N-soliton solution 
was constructed in Wronskian form \cite{2HS-10}. Furthermore,
Matsuno constructed the N-soliton solution of the 2-HS equation (\ref{Feng}) for $\mu=1$ by taking the short-wave limit 
of the N-soliton solution of the two-component Camassa--Holm (2-CH) equation \cite{2HS-11}.
More specifically, Theorem 2.1 of Lou, Feng, and Yao \cite{2HS-10} can be written in the present notation as 
\begin{eqnarray}
\left\{
\begin{array}{ll}
\left(\frac{1}{2}D_{X}D_{T}-1\right)f\cdot f=-gh,\\
\\
\left(\frac{1}{c}D_{T}+1\right)g\cdot h=f^{2},\\
\\
\left(D_{X}D_{T}+\frac{2}{c}D_{T}+cD_{X}\right)g\cdot h=0.
\end{array}
\right.
\label{LFY-bilinear}
\end{eqnarray}
Since the bilinear system (\ref{LFY-bilinear}) is invariant under $(c,g,h)\rightarrow (-c,h,g)$, the restriction $c>0$ entails no loss of generality.
This formulation uses three tau-functions $f$, $g$, and $h$, together with the dependent variable transformation
\begin{equation}
u=(\log f)_{TT},\qquad \rho=\frac{f^{2}}{gh},
\label{LFY-transform}
\end{equation}
and the hodograph transformation
\begin{equation}
x=X-(\log f)_{T},\qquad t=T.
\label{LFY-hodo}
\end{equation}

As mentioned above, the bilinear form obtained in the previous study \cite{2HS-10} 
coincides with the conventional bilinear form of the CH equation, 
and the integrable semi-discretization of the CH equation (\ref{CH}) has already been successfully established \cite{maruno1}.
However, an integrable semi-discrete analogue that preserves the two-component structure of the 2-HS equation has not yet been constructed. To address this problem, it is useful to introduce a new bilinear representation 
distinct from the conventional CH-type formulation used in \cite{2HS-10}. 
In this paper, we derive a new bilinear form of the 2-HS equation and use this new formulation, together with Hirota's 
discretization method and the discrete hodograph transformation, to construct an integrable semi-discrete analogue of (\ref{2HS}). 
In contrast to the CH-type formulation (\ref{LFY-bilinear})--(\ref{LFY-hodo}), our approach is based on a two-tau formulation 
with only $f$ and $g$, and after reduction it leads to the bilinear equations
\begin{eqnarray}
\left\{
\begin{array}{ll}
(D_{T}^{2}+cD_{T})g\cdot f=0,\\
(D_{X}(D_{T}^{2}+cD_{T})-4D_{T})g\cdot f=0,
\end{array}
\right.
\label{new-bilinear-intro}
\end{eqnarray}
together with the dependent variable transformation
\begin{equation}
\psi=\frac{g}{f},\qquad u=(\log f)_{TT},\qquad \rho=\frac{1}{1-(\log f)_{XT}},
\label{new-transform-intro}
\end{equation}
and the hodograph transformation
\begin{equation}
dX=\rho\,dx+\rho u\,dt,\qquad dT=dt.
\label{new-hodo-intro}
\end{equation}
Thus, the difference is not only in the bilinear equations themselves, but also
in the underlying tau-function structure and the associated dependent variable transformation. 
From the viewpoint of Sato theory, the integrability of the proposed semi-discrete system is built into its construction: 
the determinant identities underlying the bilinear equations are Pl\"ucker relations 
for the Wronskian and Casoratian tau functions arising from the integrable hierarchy, 
and the pseudo 2-reduction is imposed at the tau-function level. The discrete hodograph transformation 
then carries this structure to the semi-discrete system in the physical variables. 
The principal novelty of the present work is therefore the construction, through the new two-tau formulation 
and the discrete hodograph transformation, of the first integrable semi-discrete analogue of (\ref{2HS}) 
that preserves its two-component structure. 
The Lax pair is presented as an additional direct verification of the inherited integrability in the physical variables, 
rather than as the primary basis for the construction.

This paper is structured as follows.
In Section \ref{2}, we derive a new bilinear formulation of the 2-HS equation, 
distinct from the conventional bilinear form of Lou, Feng, and Yao \cite{2HS-10}, 
through a pseudo 2-reduction and show that it leads to the 2-HS equation via a hodograph transformation. We also construct its N-soliton solution in Wronskian form and present examples of one- and two-soliton interactions. 
In Section \ref{3}, we discretize this new bilinear formulation in the spatial direction, 
derive an integrable semi-discrete analogue of the 2-HS equation, construct its N-soliton solution in Casoratian form, 
and present examples of one- and two-soliton interactions. In Section \ref{sec:lax},
we construct a matrix Lax representation for the semi-discrete system and derive its scalar counterpart. We then take the continuous limit of the matrix Lax pair and show that, after scalar reduction, it reproduces the known scalar Lax representation of the continuous 2-HS equation.
Finally, in Section \ref{4}, we conclude the paper and discuss future directions.

\section{Bilinear equations and determinant solutions of the 2-HS equation}\label{2}

In this section, we introduce a new bilinear form of the 2-HS equation, 
distinct from the conventional CH-type bilinear form obtained by Lou, Feng, and Yao \cite{2HS-10}, 
and construct its N-soliton solution. 
For clarity, we emphasize here that the conventional formulation in \cite{2HS-10} is a three-tau CH-type system \eqref{LFY-bilinear} 
with the dependent variable transformation \eqref{LFY-transform}, 
whereas the present paper uses a genuinely different two-tau formulation based on $f$ and $g$ only. 
This new two-tau structure is the key feature that enables the subsequent semi-discretization in Section \ref{3} 
while preserving the two-component nature of the 2-HS equation.

Throughout this paper, $D_X$, $D_T$, and $D_S$ denote the Hirota bilinear differential operators defined by
\begin{equation}
D_X^l D_T^m D_S^r g \cdot f = \left. (\partial_X - \partial_{X'})^l (\partial_T - \partial_{T'})^m (\partial_S - \partial_{S'})^r g(X,T,S) f(X',T',S') \right|_{X'=X, T'=T, S'=S}.
\end{equation}

\begin{lemma}
The bilinear equations 
\begin{eqnarray}
\left\{
\begin{array}{ll}
(D_{T}^{2}+D_{S})g\cdot f=0,\\
(D_{X}(D_{T}^{2}+D_{S})-4D_{T})g\cdot f=0,
 \end{array}
\right.
\label{2.1}
\end{eqnarray}
admit the Wronskian solution 
\begin{equation}
f=\tau_{0},\quad g=\tau_{1},\qquad\tau_{n}=
\left|
\begin{array}{rrrr}
\varphi_{1}^{(n)}& \varphi_{1}^{(n+1)}& \cdots & \varphi_{1}^{(n+N-1)} \\
\varphi_{2}^{(n)}& \varphi_{2}^{(n+1)}& \cdots & \varphi_{2}^{(n+N-1)} \\
\vdots \quad& \vdots \quad & \ddots & \vdots \quad\quad\\
\varphi_{N}^{(n)}& \varphi_{N}^{(n+1)}& \cdots & \varphi_{N}^{(n+N-1)} \\
\end{array}
\right|,\label{2.3}
\end{equation}
where $\varphi_{i}^{(n)}=\varphi_{i}^{(n)}(X,T,S)$ satisfies the following linear dispersion relations
\begin{equation}
\frac{\partial\varphi_{i}^{(n)}}{\partial X}=\varphi_{i}^{(n+1)},\qquad\frac{\partial\varphi_{i}^{(n)}}{\partial T}=\varphi_{i}^{(n-1)},\qquad\frac{\partial\varphi_{i}^{(n)}}{\partial S}=\varphi_{i}^{(n-2)}.\label{2.4}
\end{equation}
In particular, the specific choice of $\varphi_{i}^{(n)}$ given by
\begin{eqnarray}
\varphi_{i}^{(n)}=p_{i}^{n}e^{p_{i}X+\frac{1}{p_{i}}T+\frac{1}{p_{i}^{2}}S+\xi_{i,0}}+q_{i}^{n}e^{q_{i}X+\frac{1}{q_{i}}T+\frac{1}{q_{i}^{2}}S+\eta_{i,0}},\label{phi_con}
\end{eqnarray}
satisfies the linear dispersion relations and yields the N-soliton solutions, where $\xi_{i,0}$ and $\eta_{i,0}$ are arbitrary phase constants.
\end{lemma}


\begin{proof}
For simplicity, we introduce the following notation
\begin{equation}
|m_{1},m_{2},\cdots,m_{N}|=
\left|
\begin{array}{rrrr}
\varphi_{1}^{(m_{1})}& \varphi_{1}^{(m_{2})}& \cdots & \varphi_{1}^{(m_{N})} \\
\varphi_{2}^{(m_{1})}& \varphi_{2}^{(m_{2})}& \cdots & \varphi_{2}^{(m_{N})} \\
\vdots \quad& \vdots \quad & \ddots & \vdots \quad\quad\\
\varphi_{N}^{(m_{1})}& \varphi_{N}^{(m_{2})}& \cdots & \varphi_{N}^{(m_{N})} \\
\end{array}
\right|.
\end{equation}
In this notation, $\tau_{n}$ is rewritten as 
\begin{equation}
\tau_{n}=|n,n+1,\cdots,n+N-1|.
\end{equation}
We prove that the above $\tau_{n}$ satisfies the bilinear equations (\ref{2.1}).
We obtain the following derivative formulas:

\begin{equation}
\partial_{X}\tau_{n}=|n,\cdots,n+N-2, n+N|,\label{2.71}
\end{equation}
\begin{equation}
\partial_{T}\tau_{n}=|n-1,n+1,\cdots,n+N-1|,\label{2.72}
\end{equation}
\begin{eqnarray}
\partial_{S}\tau_{n}=|n-2,n+1,\cdots,n+N-1|-|n-1,n,n+2,\cdots,n+N-1|,\label{2.73}
\end{eqnarray}
\begin{eqnarray}
\partial_{T}^{2}\tau_{n}=|n-2,n+1,\cdots,n+N-1|+|n-1,n,n+2,\cdots,n+N-1|,
\label{2.74}
\end{eqnarray}
\begin{eqnarray}
\partial_{X}\partial_{T}\tau_{n}=|n-1,n+1,\cdots,n+N-2,n+N|+|n,\cdots,n+N-1|,
\label{2.75}
\end{eqnarray}
\begin{eqnarray}
\partial_{X}\partial_{S}\tau_{n}=|n-2,n+1,\cdots,n+N-2,n+N|-|n-1,n,n+2,\cdots,n+N-2, n+N|,
\label{2.76}
\end{eqnarray}
\begingroup\footnotesize
\begin{eqnarray}
\partial_{X}\partial_{T}^{2}\tau_{n}=|n-2,n+1,\cdots,n+N-2,n+N|+|n-1,n,n+2,\cdots,n+N-2,n+N|+2|n-1,n+1,\cdots,n+N-1|,\label{2.77}
\end{eqnarray}
\endgroup

First, we show that $\tau_{n}$ satisfies the first equation of (\ref{2.1}). 
Now, we introduce an identity for the $2N\times 2N$ determinant:
\begin{equation}
\left|
\begin{array}{c:c:c:c:c:c}
n-1& n& n+2 \quad\cdots\quad n+N-1  &n+1 &\varnothing&n+N\\\hdashline
n-1& n& \varnothing&n+1 &n+2\quad\cdots\quad n+N-1 &n+N\\ 
\end{array}
\right|
=0.
\end{equation}
Applying the Laplace expansion to the left-hand side yields an algebraic bilinear identity for determinants
\begingroup\small
\begin{eqnarray}
&&|n-1,n,n+2,\cdots, n+N-1|\times|n+1,\cdots,n+N|-|n-1,n+1,\cdots,n+N-1|\times|n,n+2,\cdots,n+N|\nonumber\\
&&\qquad +|n,\cdots,n+N-1|\times|n-1,n+2,\cdots,n+N|=0.
\label{det1}
\end{eqnarray}
\endgroup
Using (\ref{2.72})--(\ref{2.74}), (\ref{det1}) can be rewritten as 
\begin{equation}
(\partial_{T}^{2}-\partial_{S})\tau_{n}\times \tau_{n+1}-2\partial_{T}\tau_{n}\times\partial_{T}\tau_{n+1}+\tau_{n}\times (\partial_{T}^{2}+\partial_{S})\tau_{n+1}=0.
\label{difbilinear1}
\end{equation}
By setting $n=0, f=\tau_{0},$ and $g=\tau_{1}$ in (\ref{difbilinear1}), the first bilinear equation of (\ref{2.1}) is obtained.

Next, we show that $\tau_{n}$ satisfies the second equation of (\ref{2.1}). 
We introduce identities for the $2N\times 2N$ determinants
\begin{equation}
\left|
\begin{array}{c:c:c:c:c:c}
n+1 & n+2\quad\cdots\quad n+N-1&n+N+1& n-1&n &\varnothing\\\hdashline
n+1&\varnothing&n+N+1 & n-1&n&n+2\quad\cdots\quad n+N-1 \\ 
\end{array}
\right|
=0,
\end{equation}
and
\begingroup\small
\begin{eqnarray}
\left|
\begin{array}{c:c:c:c:c:c:c:c}
n-1& n&n+2\quad\cdots\quad n+N-2& n+N & n+1 & \varnothing & n+N-1 &  \varnothing \\\hdashline
n-1&  n& \varnothing& \varnothing &n+1& n+2 \quad\cdots\quad n+N-2 &n+N-1&n+N \\ 
\end{array}
\right|=0.
\end{eqnarray}
\endgroup
Applying the Laplace expansion to the left-hand sides yields algebraic bilinear identities for determinants
{\small
\begin{equation}
\begin{aligned}
&|n+1,\cdots,n+N-1,n+N+1|\times|n-1,n,n+2,\cdots,n+N-1|\\
&\qquad -|n,n+2,\cdots,n+N-1,n+N+1|\times|n-1,n+1,\cdots, n+N-1|\\
&\qquad +|n,\cdots,n+N-1|\times|n-1,n+2,\cdots,n+N-1,n+N+1|=0,
\end{aligned}
\label{2.20}
\end{equation}
and
\begin{equation}
\begin{aligned}
&|n-1,n,n+2,\cdots,n+N-2,n+N|\times|n+1,\cdots,n+N|\\
&\qquad -|n-1,n+1,\cdots,n+N-2,n+N|\times|n,n+2,\cdots,n+N|\\
&\qquad +|n-1,n+2,\cdots,n+N|\times|n,\cdots,n+N-2,n+N|=0.
\end{aligned}
\label{2.21}
\end{equation}
}
Subtracting (\ref{2.21}) from (\ref{2.20}) and using (\ref{2.71})--(\ref{2.77}), we obtain 
\begingroup\footnotesize
\begin{eqnarray}
&&4(\tau_{n+1}\times \partial_{T}\tau_{n}-\tau_{n}\times \partial_{T}\tau_{n+1})-\partial_{X}\tau_{n}\times(\partial_{S}\tau_{n+1}+\partial_{T}^{2}\tau_{n+1})
-\partial_{X}\tau_{n+1}\times(\partial_{S}\tau_{n}-\partial_{T}^{2}\tau_{n})+\tau_{n+1}\times(\partial_{X}\partial_{S}\tau_{n}-\partial_{X}\partial_{T}^{2}\tau_{n})\nonumber\\
&&\qquad +\tau_{n}\times(\partial_{X}\partial_{S}\tau_{n+1}+\partial_{X}\partial_{T}^{2}\tau_{n+1})+2(\partial_{T}\tau_{n+1}\times \partial_{X}\partial_{T}\tau_{n}-\partial_{T}\tau_{n}\times\partial_{X}\partial_{T}\tau_{n+1})=0.
\label{difbilinear2}
\end{eqnarray}
\endgroup
By setting $n=0,f=\tau_{0}$, and $g=\tau_{1}$ in (\ref{difbilinear2}), the second bilinear equation of (\ref{2.1}) is obtained.
\end{proof}

To obtain the soliton solutions of the 2-HS equation, we impose the pseudo 2-reduction on the Wronskian solution in Lemma 1 by requiring
\begin{eqnarray}
\frac{1}{p_{i}}+\frac{1}{q_{i}}=c, \label{2ps}
\end{eqnarray}
for each pair of soliton parameters $p_{i}$ and $q_{i}$ in (\ref{phi_con}). 
With $f=\tau_{0}$ and $g=\tau_{1}$ as in Lemma 1, this reduction yields the following bilinear identities:
\begin{eqnarray}
D_{S}g\cdot f&=&cD_{T}g\cdot f,\nonumber\\
D_{X}D_{S}g\cdot f&=&cD_{X}D_{T}g\cdot f.\label{2con}
\end{eqnarray}
Using these identities in the bilinear equations (\ref{2.1}), we obtain the following bilinear equations 
\begin{eqnarray}
\left\{
\begin{array}{ll}
(D_{T}^{2}+cD_{T})g\cdot f=0,\\
 (D_{X}(D_{T}^{2}+cD_{T})-4D_{T})g\cdot f=0.
 \end{array}
\right.
\nonumber
\end{eqnarray}
After obtaining these reduced bilinear equations, 
the $S$-dependence can be absorbed into the phase constants, and we fix $S=0$ in the reduced $\tau$-functions.


\begin{theorem}
\label{thm:continuous-2hs}
The 2-HS equation 
\begin{eqnarray}
\left\{
\begin{array}{ll}
u_{tx}-uu_{xx}-\displaystyle\frac{1}{2}u_{x}^{2}-4u-\displaystyle\frac{c^{2}}{2}\rho^{2}+\displaystyle\frac{c^{2}}{2}=0,\\
\rho_{t}=(\rho u)_{x},\nonumber
 \end{array}
\right.
\end{eqnarray}
is derived from a set of bilinear equations 
\begin{eqnarray}
\left\{
\begin{array}{ll}
(D_{T}^{2}+cD_{T})g\cdot f=0,\\
 (D_{X}(D_{T}^{2}+cD_{T})-4D_{T})g\cdot f=0,
 \end{array}
\right.
\label{2.17}
\end{eqnarray}
through the hodograph transformation 
\begin{equation}
dX=\rho dx+\rho u dt,\quad dT=dt,\label{hd}
\end{equation}
and the dependent variable transformation 
\begin{equation}
\psi=\frac{g}{f},\quad u=(\log{f})_{TT},\quad\rho=\frac{1}{1-(\log{f})_{XT}}.\label{2.24}
\end{equation}
\end{theorem}


\begin{proof}
Dividing the bilinear equations (\ref{2.17}) by $f^{2}$, we obtain 
\begin{equation}
\left\{
\begin{array}{ll}
\displaystyle\frac{D_{T}^{2}g\cdot f}{f^{2}}+c\displaystyle\frac{D_{T}g\cdot f}{f^{2}}=0,\\
\displaystyle\frac{D_{X}D_{T}^{2}g\cdot f}{f^{2}}+c\displaystyle\frac{D_{X}D_{T}g\cdot f}{f^{2}}-4\displaystyle\frac{D_{T}g\cdot f}{f^{2}}=0,
\end{array}
\right.
\end{equation}
i.e.,
\begin{equation}
\left\{
\begin{array}{ll}
\left(\displaystyle\frac{g}{f}\right)_{TT}+2\left(\log{f}\right)_{TT}\displaystyle\frac{g}{f}+c\left(\displaystyle\frac{g}{f}\right)_{T}=0,\\
\left(\displaystyle\frac{g}{f}\right)_{XTT}+2\left(\log{f}\right)_{TT}\left(\displaystyle\frac{g}{f}\right)_{X}+4\left(\log{f}\right)_{XT}\left(\displaystyle\frac{g}{f}\right)_{T}+c\left(\left(\displaystyle\frac{g}{f}\right)_{XT}+2\left(\log{f}\right)_{XT}\left(\displaystyle\frac{g}{f}\right)\right)-4\left(\displaystyle\frac{g}{f}\right)_{T}=0.\label{2.32}
\end{array}
\right.
\end{equation}
Applying the dependent variable transformation (\ref{2.24}), we can rewrite (\ref{2.32}) as 
\begin{equation}
\left\{
\begin{array}{ll}
\psi_{TT}+2u\psi+c\psi_{T}=0,\\
\psi_{XTT}+2u\psi_{X}+4\left(1-\displaystyle\frac{1}{\rho}\right)\psi_{T}+c\left(\psi_{XT}+2\left(1-\displaystyle\frac{1}{\rho}\right)\psi\right)-4\psi_{T}=0.\end{array}
\right.
\label{2.27}
\end{equation}
Next, we eliminate $\psi$ from (\ref{2.27}).
Subtracting the $X$-derivative of the first equation of (\ref{2.27}) from the second equation of (\ref{2.27}),  we have
\begin{equation}
\frac{\psi_{T}}{\psi}=-\frac{\rho}{2}u_{X}+\frac{1}{2}\rho c-\frac{1}{2}c.\label{2.29}
\end{equation}
From the first equation of (\ref{2.27}), we obtain 
\begin{equation}
\left(\frac{\psi_{T}}{\psi}\right)_{T}+\left(\frac{\psi_{T}}{\psi}\right)^{2}+2u+c\frac{\psi_{T}}{\psi}=0.\label{2.30}
\end{equation}
Substituting (\ref{2.29}) into (\ref{2.30}) and rearranging, we have
\begin{equation}
\rho_{T}u_{X}+\rho u_{XT}-\frac{1}{2}(\rho u_{X})^{2}+c(\rho^{2}u_{X}-\rho_{T})-4u-\frac{c^{2}}{2}\rho^{2}+\frac{c^{2}}{2}=0.\label{2.31}
\end{equation}
Now, we obtain the conservation law 
\begin{equation}
u_{X}=\left(1-\frac{1}{\rho}\right)_{T}, \label{conlow}
\end{equation}
from the dependent variable transformation (\ref{2.24}) and it leads to
\begin{equation}
\rho_{T}=\rho^{2}u_{X}. \label{2.33}
\end{equation} 
Substituting (\ref{2.33}) into (\ref{2.31}), we have 
\begin{equation}
(\rho u_{X})_{T}-\frac{1}{2}(\rho u_{X})^{2}-4u-\frac{c^{2}}{2}\rho^{2}+\frac{c^{2}}{2}=0.\label{2.34}
\end{equation}
Applying the differential law
\begin{equation}
\partial_{X}=\frac{1}{\rho}\partial_{x},\quad\partial_{T}=\partial_{t}-u\partial_{x},\label{diflaw}
\end{equation}
which is derived from the hodograph transformation (\ref{hd}), to (\ref{conlow}) and (\ref{2.34}), we obtain
\begin{eqnarray}
\left\{
\begin{array}{ll}
u_{tx}-uu_{xx}-\displaystyle\frac{1}{2}u_{x}^{2}-4u-\displaystyle\frac{c^{2}}{2}\rho^{2}+\displaystyle\frac{c^{2}}{2}=0,\\
\rho_{t}=(\rho u)_{x},\nonumber
 \end{array}
\right.
\label{2.37}
\end{eqnarray}
which is exactly the 2-HS equation (\ref{2HS}).
\end{proof}


\begin{example}
The reduced Wronskian $\tau$-functions yield the explicit one- and two-soliton solutions of 
the 2-HS equation through the dependent variable transformation 
\begin{equation}
u=(\log{f})_{TT},\quad \rho=\frac{1}{1-(\log f)_{XT}},
\end{equation}
and the hodograph transformation 
\begin{equation}
x=X-(\log{f})_{T},\quad t=T.
\end{equation}

\vspace{1cm}
\noindent
{\bf One-soliton:}
The $\tau$-functions for the one-soliton solution are given below.
\begin{eqnarray}
f=e^{\xi_{1}}+e^{\eta_{1}}\propto1+e^{\xi_{1}-\eta_{1}},\quad
g=p_{1}e^{\xi_{1}}+q_{1}e^{\eta_{1}}\propto1+\frac{p_{1}}{q_{1}}e^{\xi_{1}-\eta_{1}},
\label{2.47}
\end{eqnarray}
where 
\begin{equation}
\xi_{1}=p_{1}X+\frac{1}{p_{1}}T+\xi_{1,0},\quad 
\eta_{1}=q_{1}X+\frac{1}{q_{1}}T+\eta_{1,0}.
\end{equation}
The parameters $p_{1}$ and $q_{1}$ satisfy $\frac{1}{p_{1}}+\frac{1}{q_{1}}=c$.
\\
This yields the one-soliton solution of the 2-HS equation
\begin{equation}
u=\left(\frac{p_{1}-q_{1}}{2p_{1}q_{1}}\right)^2 \mathrm{sech}^{2}\frac{\theta}{2},
\qquad\rho=\left(1+\frac{(p_{1}-q_{1})^{2}}{4p_{1}q_{1}} \mathrm{sech}^{2}\frac{\theta}{2}\right)^{-1},
\end{equation}
\begin{equation}
x=X-\frac{1+ \mathrm{tanh}\frac{\theta}{2}}{p_{1}\left(1+ \mathrm{tanh}\frac{\theta-\theta^{\prime}}{2}\right)},\quad t=T,
\end{equation}
\begin{equation}
\theta=\left(p_{1}-q_{1}\right)X+\left(\frac{1}{p_{1}}-\frac{1}{q_{1}}\right)T+\xi_{1,0}-\eta_{1,0},\quad \theta^{\prime}=\log{\frac{p_{1}}{q_{1}}},
\label{2.54}
\end{equation}
where $q_{1}=\frac{p_{1}}{-1 + c p_{1}}$.
The amplitude of the $u$-component soliton is $\left(\frac{cp_{1}-2}{2p_{1}}\right)^{2}$, and 
its velocity in the $(X,T)$ plane is $\frac{cp_{1}-1}{p_{1}^{2}}$. 

Figure~\ref{fig2.1} shows the one-soliton solution for $t=0$.
The soliton moves to the right along the $x$-axis as time progresses.


\begin{figure}[h]
\centering
         \includegraphics[keepaspectratio, scale=0.5]{
         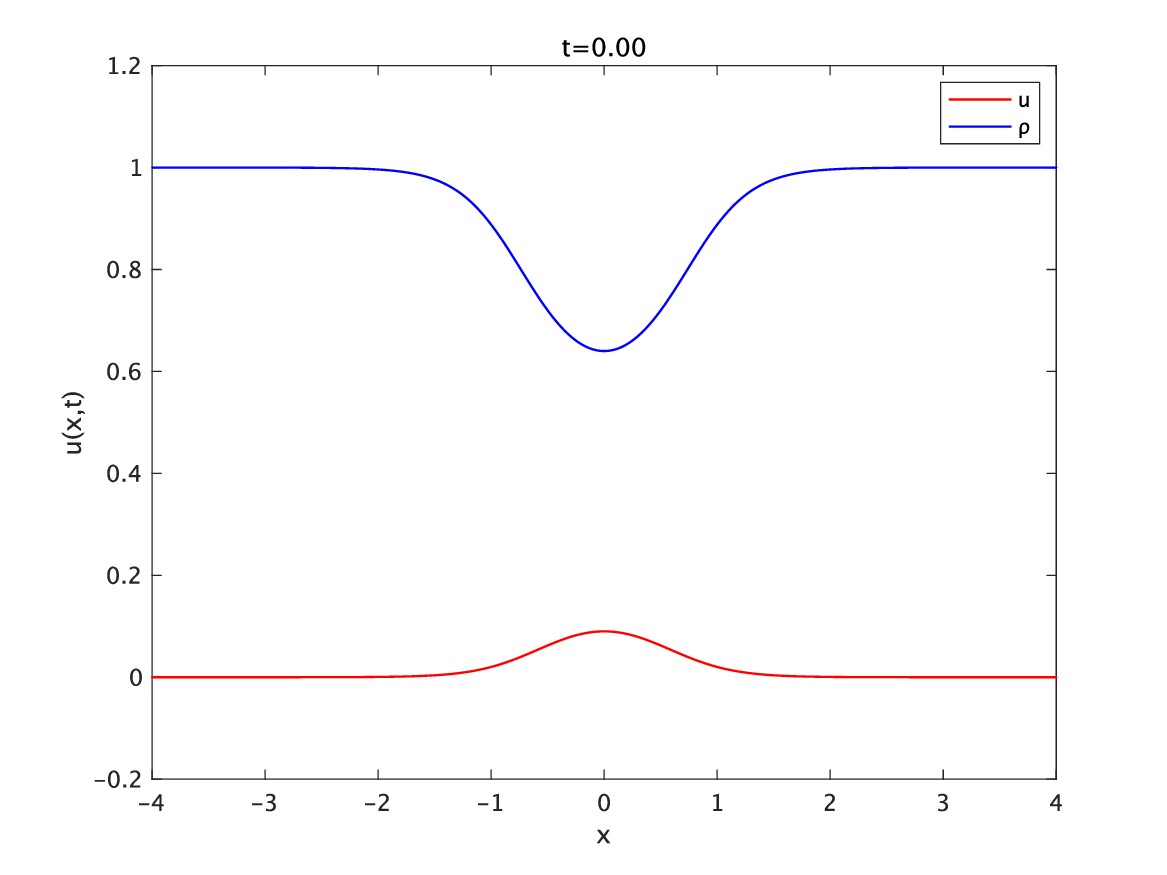}
         \caption{The one-soliton solution of the 
         two-component Hunter--Saxton equation (red (lower) line: $u$, blue (upper) line: $\rho$). 
         The parameters are $p_{1}=5, c=1, \eta_{1,0}=\xi_{1,0}=0$.}
         \label{fig2.1}

\end{figure}

\noindent	
{\bf Two-soliton:}
The $\tau$-functions for the two-soliton solution are given below.
In writing the proportional forms, we first remove common gauge
factors, which do not affect $u$ or $\rho$ because these variables are defined
by logarithmic derivatives of the $\tau$-functions. The nonzero coefficients
of the two one-soliton exponential terms are then absorbed into shifts of the
arbitrary phase constants, whereas the resulting interaction coefficient is
retained.
\begin{eqnarray}
&&f=(p_{2}-p_{1})e^{\xi_{1}+\xi_{2}}+(q_{2}-p_{1})e^{\xi_{1}+\eta_{2}}+(p_{2}-q_{1})e^{\xi_{2}+\eta_{1}}+(q_{2}-q_{1})e^{\eta_{1}+\eta_{2}}\nonumber\\
&&\qquad \propto 1+e^{\xi_{1}-\eta_{1}}+e^{\xi_{2}-\eta_{2}}+\frac{(p_{2}-p_{1})(q_{2}-q_{1})}{(q_{2}-p_{1})(p_{2}-q_{1})}e^{\xi_{1}+\xi_{2}-\eta_{1}-\eta_{2}},
\end{eqnarray}
\begin{eqnarray}
&&g=p_{1}p_{2}(p_{2}-p_{1})e^{\xi_{1}+\xi_{2}}+p_{1}q_{2}(q_{2}-p_{1})e^{\xi_{1}+\eta_{2}}+p_{2}q_{1}(p_{2}-q_{1})e^{\xi_{2}+\eta_{1}}
+q_{1}q_{2}(q_{2}-q_{1})e^{\eta_{1}+\eta_{2}}\nonumber\\
&&\qquad \propto 1+\frac{p_{1}}{q_{1}}e^{\xi_{1}-\eta_{1}}+\frac{p_{2}}{q_{2}}e^{\xi_{2}-\eta_{2}}
+\frac{p_{1}p_{2}(p_{2}-p_{1})(q_{2}-q_{1})}{q_{1}q_{2}(q_{2}-p_{1})(p_{2}-q_{1})}e^{\xi_{1}+\xi_{2}-\eta_{1}-\eta_{2}},
\end{eqnarray}
where
\begin{equation}
\xi_{i}=p_{i}X+\frac{1}{p_{i}}T+\xi_{i,0},\quad
\eta_{i}=q_{i}X+\frac{1}{q_{i}}T+\eta_{i,0} \quad{\rm for}\quad i=1,2.
\end{equation}
The parameters $p_{i}$ and $q_{i}$ satisfy $\frac{1}{p_{i}}+\frac{1}{q_{i}}=c$ for $i=1,2$.

Figure~\ref{fig2.2} shows the two-soliton interaction of $u$ and
$\rho$. Both solitons move to the right along the $x$-axis as time progresses.


\begin{figure}[h]
  \begin{tabular}{cc}
      \begin{minipage}[t]{0.45\hsize}
              \includegraphics[width=\linewidth]{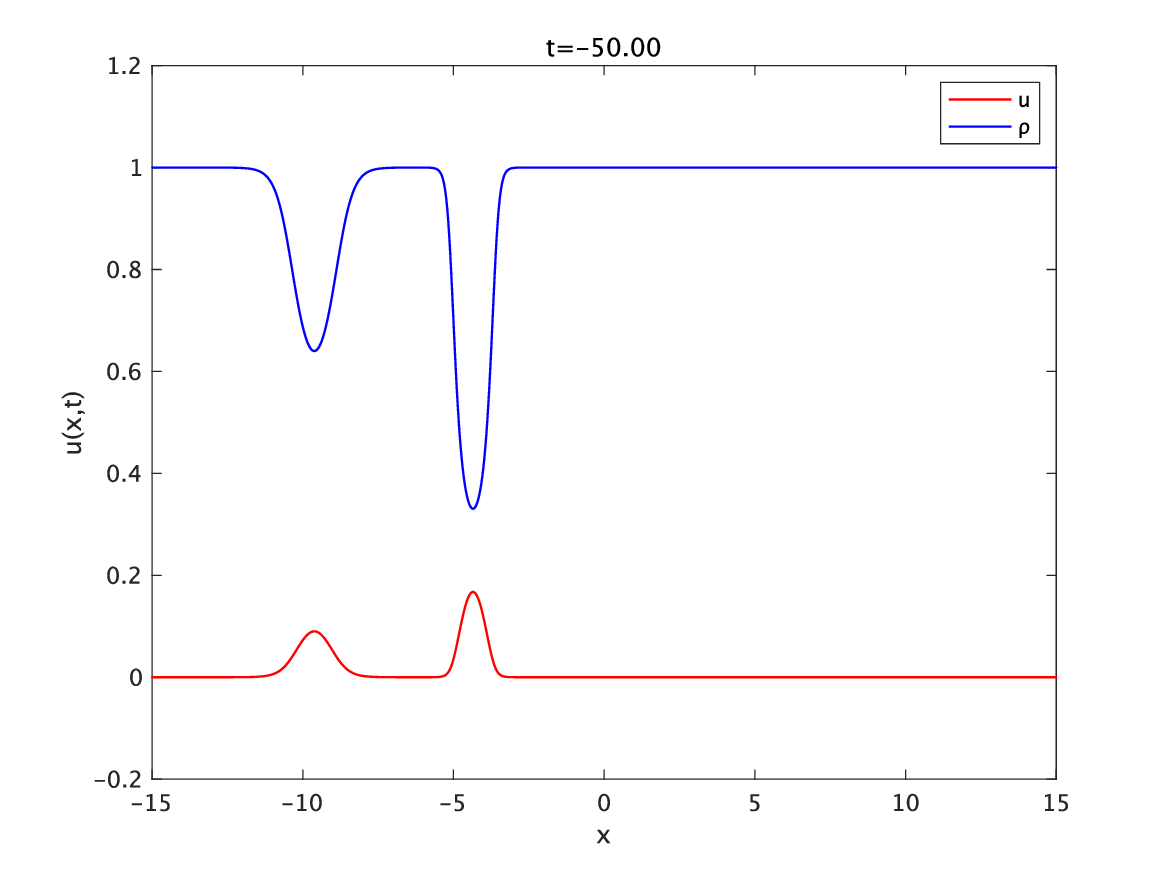}
        \captionsetup{labelformat=empty,labelsep=none}
      \end{minipage}
                \begin{minipage}[t]{0.45\hsize}
        \includegraphics[width=\linewidth]{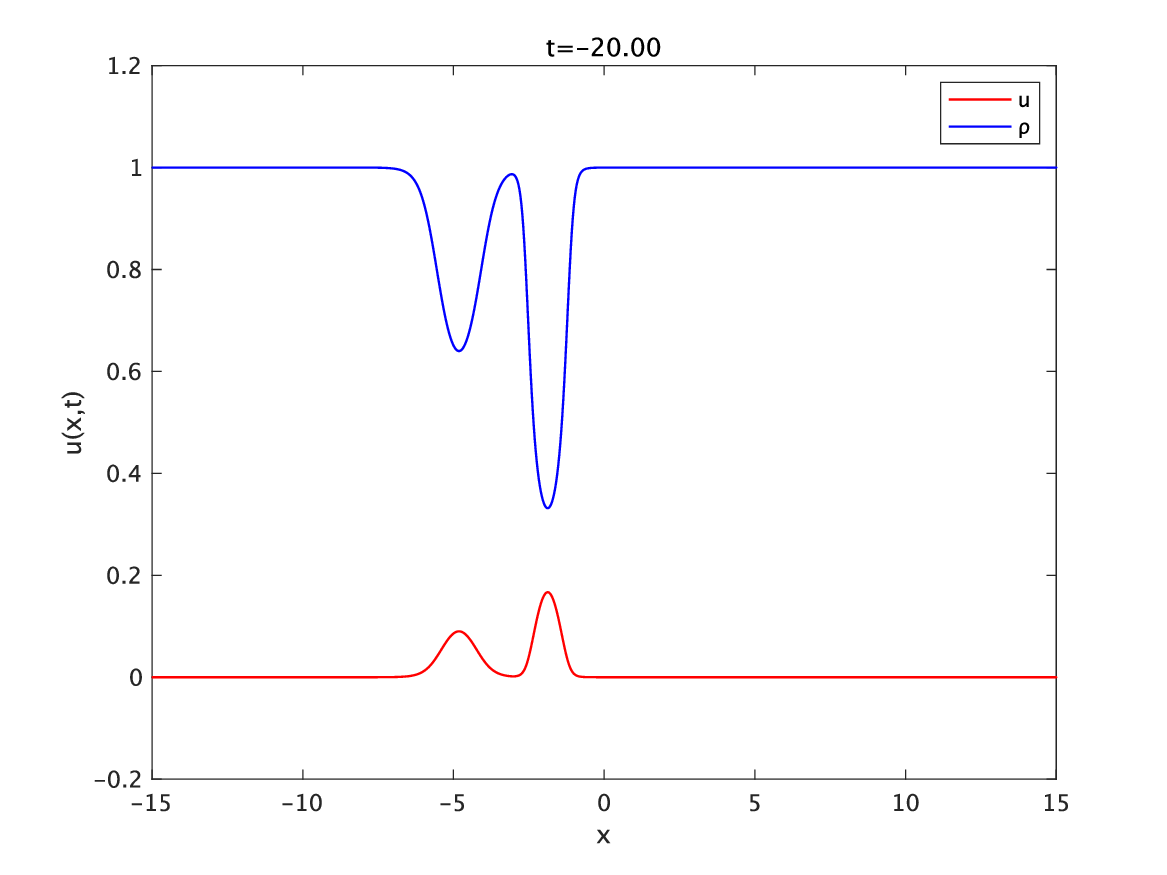}
        \captionsetup{labelformat=empty,labelsep=none}
      \end{minipage}
                     \end{tabular}
\\

 \begin{tabular}{cc}
      \begin{minipage}[t]{0.45\hsize}
        \includegraphics[width=\linewidth]{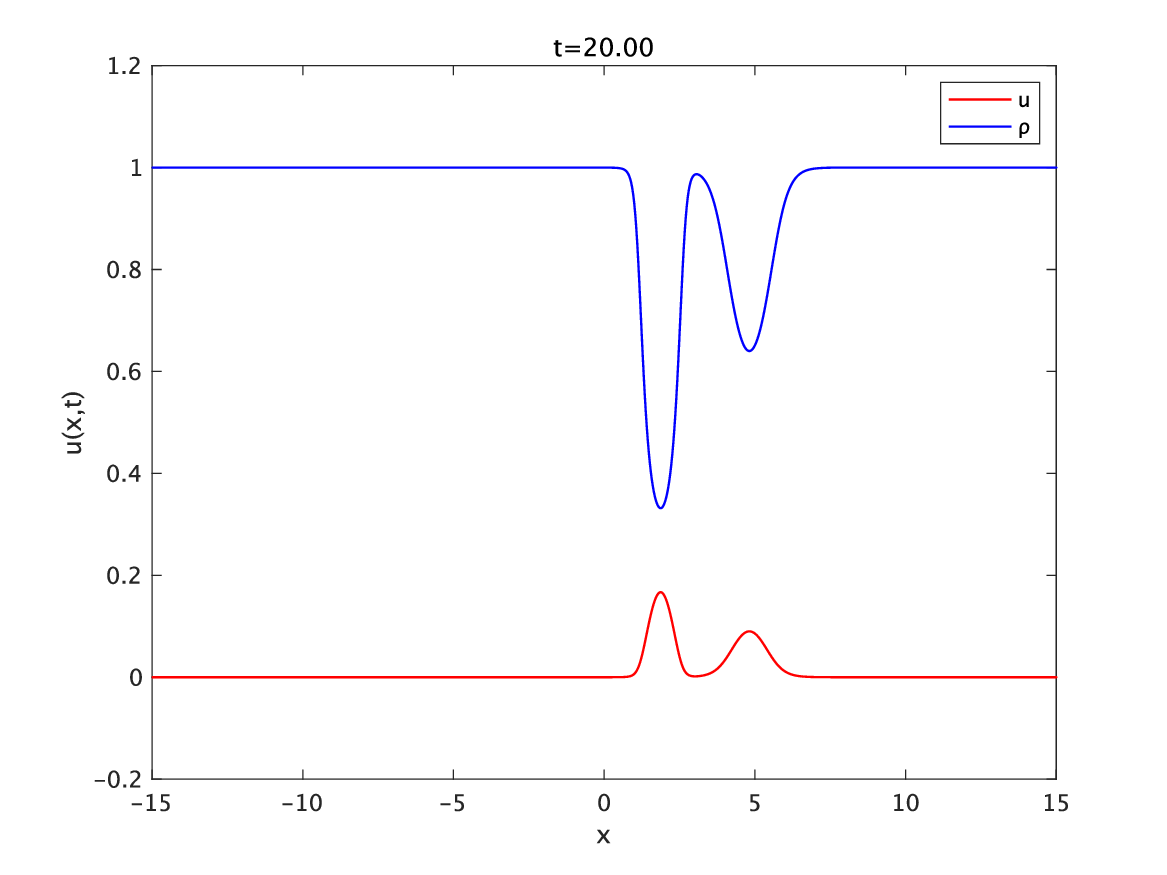}
        \captionsetup{labelformat=empty,labelsep=none}
      \end{minipage}

      \begin{minipage}[t]{0.45\hsize}
        \includegraphics[width=\linewidth]{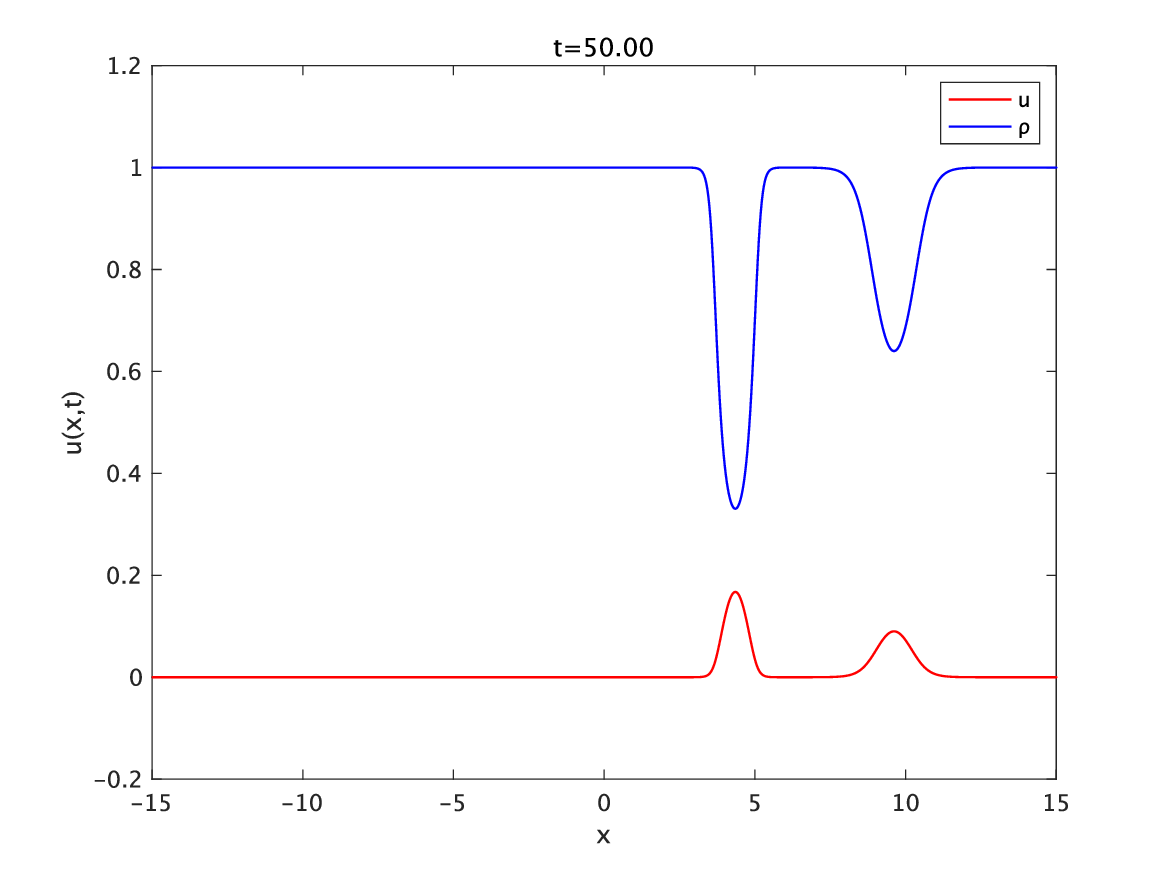}
        \captionsetup{labelformat=empty,labelsep=none}
      \end{minipage} 
         \end{tabular}
         \caption{The two-soliton interaction of the 2-HS equation 
         (red (lower) line: $u$, blue (upper) line: $\rho$). 
         The parameters are $p_{1}=1.1, p_{2}=1.25, c=1, \xi_{1,0}=\frac{1}{2}
         \log{((q_{1} - p_{2}) (q_{1}- q_{2}))}, \eta_{1,0}= \frac{1}{2} 
         \log{((p_{2} - p_{1}) (q_{2} - p_{1}))} + \pi \mathrm{i}, \xi_{2,0}=\frac{1}{2}\log{((q_{2} - p_{1}) (q_{1} - q_{2}))}, 
         \eta_{2,0}=\frac{1}{2}\log{((p_{2} - p_{1}) (q_{1} - p_{2}))}$.}
         \label{fig2.2}
         
  \end{figure}

\end{example}

\begin{remark}
In the case of the one-soliton solution with $c>0$, $u$ is a smooth soliton 
if $p_{1} > \frac{1}{c}$, whereas $u$ is a loop soliton if $p_{1}<\frac{1}{c}$. 
However, in the latter case,  $\rho$ diverges. Therefore, we restrict the parameters so that both $u$ and $\rho$ remain smooth, 
leading to the condition $p_{1} > \frac{1}{c}$. Similarly, in the case of the two-soliton solution, 
the parameters are restricted to $p_{1}, p_{2}>\frac{1}{c}$.
\end{remark}

\begin{remark}
In the limit $c\rightarrow0$, (\ref{2.17}) reduces to the bilinear equations 
of the HS equation. For the loop-soliton branch, taking the same limit yields 
the cuspon solution of the HS equation. Details are described in the Appendix.
\end{remark}

\section{An integrable semi-discrete analogue of the 2-HS equation}
\label{3}


\begin{lemma}
The bilinear equations 
\begin{eqnarray}
\left\{
\begin{array}{ll}
(D_{T}^{2}+D_{S})g_{k}\cdot f_{k}=0\\
 (D_{T}^{2}+D_{S}-2a D_{T})g_{k}\cdot f_{k-1}=0
 \end{array}
\right.
\label{3.1}
\end{eqnarray}
where $k$ is a discrete site index and $a$ is a lattice parameter, admit the Casoratian solution 
\begin{equation}
f_{k}=\tau_{0}(k),\quad g_{k}=\tau_{1}(k),\qquad
\tau_{n}(k)=
\left|
\begin{array}{rrrr}
\varphi_{1}^{(n)}(k)& \varphi_{1}^{(n+1)}(k)& \cdots & \varphi_{1}^{(n+N-1)}(k) \\
\varphi_{2}^{(n)}(k)& \varphi_{2}^{(n+1)}(k)& \cdots & \varphi_{2}^{(n+N-1)}(k) \\
\vdots& \vdots& \ddots & \vdots \\
\varphi_{N}^{(n)}(k)& \varphi_{N}^{(n+1)}(k)& \cdots & \varphi_{N}^{(n+N-1)}(k)
\label{3.3}
\end{array}
\right|,
\end{equation}
where $\varphi_{i}^{(n)}(k)=\varphi_{i}^{(n)}(k,T,S)$ satisfies 
the following linear dispersion relations
\begin{equation}
\frac{\varphi_{i}^{(n)}(k)-\varphi_{i}^{(n)}(k-1)}{a}=\varphi_{i}^{(n+1)}(k),\label{3.4}
\end{equation}
and
\begin{equation}
\frac{\partial\varphi_{i}^{(n)}(k)}{\partial T}=\varphi_{i}^{(n-1)}(k),\quad
\frac{\partial\varphi_{i}^{(n)}(k)}{\partial S}=\varphi_{i}^{(n-2)}(k).
\end{equation}
In particular, the specific choice of $\varphi_{i}^{(n)}(k)$ given by
\begin{equation}
\varphi_{i}^{(n)}(k)=p_{i}^{n}(1-ap_{i})^{-k}e^{\frac{1}{p_{i}}T+\frac{1}{p_{i}^{2}}S+\xi_{i,0}}
+q_{i}^{n}(1-aq_{i})^{-k}e^{\frac{1}{q_{i}}T+\frac{1}{q_{i}^{2}}S+\eta_{i,0}},\label{phi_dis}
\end{equation}
satisfies the linear dispersion relations and yields the N-soliton solution.
\end{lemma}


\begin{proof}
For simplicity, we introduce the following notation
\begin{equation}
|m_{1\,k},m_{2\,k},\cdots,m_{N\,k}|=
\left|
\begin{array}{rrrr}
\varphi_{1}^{(m_{1})}(k)& \varphi_{1}^{(m_{2})}(k)& \cdots & \varphi_{1}^{(m_{N})}(k) \\
\varphi_{2}^{(m_{1})}(k)& \varphi_{2}^{(m_{2})}(k)& \cdots & \varphi_{2}^{(m_{N})}(k) \\
\vdots \quad& \vdots \quad & \ddots & \vdots \quad\quad\\
\varphi_{N}^{(m_{1})}(k)& \varphi_{N}^{(m_{2})}(k)& \cdots & \varphi_{N}^{(m_{N})}(k) \\
\end{array}
\right|.
\end{equation}
In this notation, $\tau_{n}(k)$ is rewritten as 
\begin{equation}
\tau_{n}(k)=|n_{k},n+1_{k},\cdots, n+N-1_{k}|.
\end{equation}
The proof of the first equation of (\ref{3.1}) is the same as for the continuous case. 
We prove that $\tau_{n}(k)$ satisfies the second equation of (\ref{3.1}).
\begin{eqnarray}
\partial_{T}\tau_{n}(k)=|n-1_{k},n+1_{k},\cdots,n+N-1_{k}|,\label{3.111}
\end{eqnarray}
\begin{eqnarray}
\partial_{S}\tau_{n}(k)=|n-2_{k},n+1_{k},\cdots,n+N-1_{k}|-|n-1_{k},n_{k},n+2_{k},\cdots,n+N-1_{k}|,\label{3.112}
\end{eqnarray}
\begin{eqnarray}
\partial_{T}^{2}\tau_{n}(k)=|n-2_{k},n+1_{k},\cdots,n+N-1_{k}|+|n-1_{k},n_{k},n+2_{k},\cdots,n+N-1_{k}|.\label{3.113}
\end{eqnarray}

Let us introduce an identity for the $2N\times2N$ determinant
\begingroup\footnotesize
\begin{eqnarray}
\left|
\begin{array}{c:c:c:c:c:c}
n_{k-1}&n+1_{k-1}&n+2_{k-1}\quad\cdots\quad n+N-1_{k-1}&n-1_{k-1} &\varnothing&n+N_{k}  \\\hdashline
n_{k-1}&n+1_{k-1}&\varnothing&n-1_{k-1}&n+2_{k-1}\quad\cdots\quad n+N-1_{k-1}&n+N_{k} \\
\end{array}
\right|=0.
\end{eqnarray}
\endgroup
Applying the Laplace expansion to the left-hand side yields an algebraic bilinear identity for determinants
\begin{equation}
\begin{aligned}
&|n_{k-1},\cdots,n+N-1_{k-1}|\times|n-1_{k-1},n+2_{k-1},\cdots,n+N_{k}|\\
&\qquad -|n-1_{k-1},n+1_{k-1},\cdots,n+N-1_{k-1}|\times|n_{k-1},n+2_{k-1},\cdots,n+N_{k}|\\
&\qquad +|n-1_{k-1},n_{k-1},n+2_{k-1},\cdots,n+N-1_{k-1}|\times|n+1_{k-1},\cdots,n+N_{k}|=0.
\end{aligned}
\label{3.88}
\end{equation}
By imposing the linear dispersion relation (\ref{3.4}), (\ref{3.88}) reduces to
\begin{equation}
\begin{aligned}
&|n_{k-1},\cdots,n+N-1_{k-1}|\times|n-1_{k},n+2_{k},\cdots,n+N_{k}|\\
&\qquad -a|n_{k-1},\cdots,n+N-1_{k-1}|\times|n_{k},n+2_{k},\cdots,n+N_{k}|\\
&\qquad -|n-1_{k-1},n+1_{k-1},\cdots,n+N-1_{k-1}|\times|n_{k},n+2_{k},\cdots,n+N_{k}|\\
&\qquad +a|n-1_{k-1},n+1_{k-1},\cdots,n+N-1_{k-1}|\times|n+1_{k},\cdots,n+N_{k}|\\
&\qquad +|n-1_{k-1},n_{k-1},n+2_{k-1},\cdots,n+N-1_{k-1}|\times|n+1_{k},\cdots,n+N_{k}|=0.
\end{aligned}
\label{3.9}
\end{equation}
From (\ref{3.111})--(\ref{3.113}) and (\ref{3.9}), we obtain
\begin{equation}
\begin{aligned}
&\partial_{T}^{2}\tau_{n}(k-1)\times \tau_{n+1}(k)+\tau_{n}(k-1)\times\partial_{T}^{2}\tau_{n+1}(k)-2\partial_{T}\tau_{n}(k-1)\times \partial_{T}\tau_{n+1}(k)\\
&\qquad +\partial_{S}\tau_{n+1}(k)\times \tau_{n}(k-1)-\partial_{S}\tau_{n}(k-1)\times \tau_{n+1}(k)\\
&\qquad -2a(\partial_{T}\tau_{n+1}(k)\times \tau_{n}(k-1)-\tau_{n+1}(k)\times\partial_{T}\tau_{n}(k-1))=0.
\end{aligned}
\label{3.13}
\end{equation}
By setting $n = 0, f_{k} = \tau_{0}(k)$ and $g_{k} = \tau_{1}(k)$ in (\ref{3.13}), the second equation of (\ref{3.1}) is obtained.

\end{proof}

To obtain the soliton solutions of the semi-discrete 2-HS equation, we impose the pseudo 2-reduction on the Casoratian solution in Lemma 2 by requiring
\begin{eqnarray}
\frac{1}{p_{i}}+\frac{1}{q_{i}}=c,
\label{2pps}
\end{eqnarray}
for each pair of soliton parameters $p_{i}$ and $q_{i}$ in (\ref{phi_dis}). With $f_{k}=\tau_{0}(k)$ and $g_{k}=\tau_{1}(k)$ as in Lemma 2, this reduction yields the following bilinear identities:
\begin{eqnarray}
D_{S}g_{k}\cdot f_{k}&=&cD_{T}g_{k}\cdot f_{k},\nonumber\\
D_{S}g_{k}\cdot f_{k-1}&=&cD_{T}g_{k}\cdot f_{k-1}.\label{2ccon}
\end{eqnarray}
Using these identities in the bilinear equations (\ref{3.1}), we obtain the following bilinear equations 
\begin{eqnarray}
\left\{
\begin{array}{ll}
(D_{T}^{2}+cD_{T})g_{k}\cdot f_{k}=0,\\
(D_{T}^{2}+cD_{T}-2a D_{T})g_{k}\cdot f_{k-1}=0.
\end{array}
\right.
\nonumber
\end{eqnarray}
After obtaining these reduced bilinear equations, 
the $S$-dependence can be absorbed into the phase constants, 
and we fix $S=0$ in the reduced $\tau$-functions.


\begin{theorem}
\label{thm:semi-discrete-2hs}
The semi-discrete analogue of the 2-HS equation
\begin{equation}
\left\{
\begin{aligned}
&\frac{d}{dT}(u_{k}-u_{k-1})
=-\frac{u_{k}-u_{k-1}}{\delta_{k}}
 \left((u_{k}-u_{k-1})+(\delta_{k}-a)(c-a-\delta_{k})\right)\\
&\qquad +\frac{1}{2\delta_{k}}
 \left((u_{k}-u_{k-1})+(\delta_{k}-a)(c-a-\delta_{k})\right)^{2}
 +4u_{k}\delta_{k}\\
&\qquad -c\left((u_{k}-u_{k-1})+(\delta_{k}-a)(c-a-\delta_{k})\right)
 -(u_{k}-u_{k-1})(2\delta_{k}-c),\\
&\frac{dx_{k}}{dT}=-u_{k},
\end{aligned}
\right.
\label{semi-2HS}
\end{equation}
is derived from a set of bilinear equations
\begin{eqnarray}
\left\{
\begin{array}{ll}
(D_{T}^{2}+cD_{T})g_{k}\cdot f_{k}=0,\\
(D_{T}^{2}+cD_{T}-2a D_{T})g_{k}\cdot f_{k-1}=0,
\end{array}
\right.
\label{semi-bilinear-system}
\end{eqnarray}
through the dependent variable transformation
\begin{equation}
\psi_{k}=\frac{g_{k}}{f_{k}},\quad u_{k}=(\log{f_{k}})_{TT},\quad \frac{1}{\rho_{k}}=1-\frac{1}{a}\left(\log{\frac{f_{k}}{f_{k-1}}}\right)_{T},\label{ddt}
\end{equation}
together with the discrete hodograph transformation
\begin{equation}
x_{k}=x_{0}(T)+\sum_{l=1}^{k}\frac{a}{\rho_{l}},
\label{hoddis}
\end{equation}
where the left-edge mesh point satisfies the consistency condition
\begin{equation}
\frac{dx_{0}(T)}{dT}=-u_{0}(T).
\label{dcondition}
\end{equation}
Here, $\delta_{k}=x_{k}-x_{k-1}$, and $u_{k}$ denotes the value of $u$ at $x_{k}$.
\end{theorem}


\begin{proof}

Dividing the first and second equations of
(\ref{semi-bilinear-system}) by $f_{k}^{2}$ and $f_{k}f_{k-1}$,
respectively, we obtain
\begin{eqnarray}
\left\{
\begin{array}{ll}
\displaystyle\frac{D_{T}^{2}g_{k}\cdot f_{k}}{f_{k}^{2}}+c\displaystyle\frac{D_{T}g_{k}\cdot f_{k}}{f_{k}^{2}}=0,\\
\displaystyle\frac{D_{T}^{2}g_{k}\cdot f_{k-1}}{f_{k}f_{k-1}}+(c-2a)\displaystyle\frac{D_{T}g_{k}\cdot f_{k-1}}{f_{k}f_{k-1}}=0.
 \end{array}
\right.
\label{3.14}
\end{eqnarray}
The equations (\ref{3.14}) are rewritten as 
\begin{eqnarray}
\left\{
\begin{array}{l}
\left(\displaystyle\frac{g_{k}}{f_{k}}\right)_{TT}+2\left(\log{f_{k}}\right)_{TT}\displaystyle\frac{g_{k}}{f_{k}}+c\left(\displaystyle\frac{g_{k}}{f_{k}}\right)_{T}=0,\\
\left(\displaystyle\frac{g_{k}}{f_{k}}\right)_{TT}+2\left(\left(\displaystyle\frac{g_{k}}{f_{k}}\right)_{T}+\displaystyle\frac{f_{k,T}g_{k}}{f_{k}^{2}}\right)
\left(\displaystyle\frac{f_{k,T}}{f_{k}}-\displaystyle\frac{f_{k-1,T}}{f_{k-1}}\right)+\displaystyle\frac{g_{k}}{f_{k}}\left(\displaystyle\frac{f_{k,TT}}{f_{k}}-\frac{f_{k,T}^{2}}{f_{k}^{2}}\right)\\
\qquad +\displaystyle\frac{g_{k}}{f_{k}}\left((\log{f_{k-1}})_{TT}+\displaystyle\frac{f_{k-1,T}^{2}}{f_{k-1}^{2}}
-\displaystyle\frac{f_{k,T}^{2}}{f_{k}^{2}}\right)+(c-2a)\left(\left(\displaystyle\frac{g_{k}}{f_{k}}\right)_{T}+\displaystyle\frac{g_{k}}{f_{k}}\left(\log{\displaystyle\frac{f_{k}}{f_{k-1}}}\right)_{T}\right)=0.\\
\end{array}
\right.
\label{85}
\end{eqnarray}
Applying the dependent variable transformation (\ref{ddt}), we can reduce (\ref{85}) to
\begin{eqnarray}
\left\{
\begin{array}{l}
\psi_{k,TT}+2u_{k}\psi_{k}+c\psi_{k,T}=0,\\
\psi_{k,TT}+2a\psi_{k,T}\left(1-\displaystyle\frac{1}{\rho_{k}}\right)+\psi_{k}(u_{k}+u_{k-1})\\
\qquad +a^{2}\psi_{k}\left(1-\displaystyle\frac{1}{\rho_{k}}\right)^{2}+(c-2a)\left(\psi_{k,T}+a\psi_{k}\left(1-\displaystyle\frac{1}{\rho_{k}}\right)\right)=0.
\end{array}
\right.
\label{3.16}
\end{eqnarray}
Subtracting the first equation from the second equation of (\ref{3.16}), we obtain
\begin{eqnarray}
&&(u_{k-1}-u_{k})\psi_{k}+2a\left(1-\frac{1}{\rho_{k}}\right)\psi_{k,T}+ca\left(1-\frac{1}{\rho_{k}}\right)\psi_{k}\nonumber\\
&&\qquad -2a\left(\psi_{k,T}+a\left(1-\frac{1}{\rho_{k}}\right)\psi_{k}\right)+a^{2}\left(1-\frac{1}{\rho_{k}}\right)^{2}\psi_{k}=0.\label{3.17}
\end{eqnarray}
Dividing (\ref{3.17}) by $\psi_{k}$ and rearranging, we have
\begin{eqnarray}
\frac{\psi_{k,T}}{\psi_{k}}=\frac{\rho_{k}}{2a}\left(u_{k-1}-u_{k}+ac\left(1-\frac{1}{\rho_{k}}\right)-a^{2}\left(1-\frac{1}{\rho_{k}^{2}}\right)\right).\label{3.18}
\end{eqnarray}
Dividing the first equation of (\ref{3.16}) by $\psi_{k}$ yields
\begin{equation}
\left(\frac{\psi_{k,T}}{\psi_{k}}\right)_{T}+\left(\frac{\psi_{k,T}}{\psi_{k}}\right)^{2}+2u_{k}+c\frac{\psi_{k,T}}{\psi_{k}}=0,\label{3.19}
\end{equation}
and substituting (\ref{3.18}) into (\ref{3.19}), we have
\begin{equation}
\begin{aligned}
&\left[\frac{\rho_{k}}{2a}\left((u_{k-1}-u_{k})+ac\left(1-\frac{1}{\rho_{k}}\right)-a^{2}\left(1-\frac{1}{\rho_{k}^{2}}\right)\right)\right]_{T}\\
&\qquad +\left[\frac{\rho_{k}}{2a}\left((u_{k-1}-u_{k})+ac\left(1-\frac{1}{\rho_{k}}\right)-a^{2}\left(1-\frac{1}{\rho_{k}^{2}}\right)\right)\right]^{2}\\
&\qquad +2u_{k}+c\left[\frac{\rho_{k}}{2a}\left((u_{k-1}-u_{k})+ac\left(1-\frac{1}{\rho_{k}}\right)-a^{2}\left(1-\frac{1}{\rho_{k}^{2}}\right)\right)\right]=0.
\end{aligned}
\label{3.20}
\end{equation}
The integral form of the hodograph transformation is given by
\begin{eqnarray}
x=\int\frac{1}{\rho(X,T)}dX=x_{0}(T)+\int_{0}^{X}\frac{1}{\rho(\bar{X},T)}d\bar{X},\label{hod}
\end{eqnarray}
where $x_{0}(T)$ is an integration function determined at the left-hand edge $(x=x_{0}(T))$.
Discretizing (\ref{hod}) in the spatial direction gives the
discrete hodograph transformation (\ref{hoddis}) stated in
Theorem~\ref{thm:semi-discrete-2hs}, where
$\rho_{k}=\rho(X_{k},T)$ with $X_{k} = k a$.

Introducing a mesh interval
\begin{equation}
\delta_{k}:=x_{k}-x_{k-1},
\end{equation}
we obtain
\begin{equation}
\delta_{k}=\frac{a}{\rho_{k}},\label{3.24}
\end{equation}
from the discrete hodograph transformation (\ref{hoddis}). Substituting (\ref{3.24}) into (\ref{3.20}), we have
\begin{equation}
\begin{aligned}
&\left[\frac{1}{2\delta_{k}}((u_{k-1}-u_{k})+(a-\delta_{k})(c-a-\delta_{k}))\right]_{T}\\
&\qquad +\left[\frac{1}{2\delta_{k}}((u_{k-1}-u_{k})+(a-\delta_{k})(c-a-\delta_{k}))\right]^{2}\\
&\qquad +2u_{k}+c\left[\frac{1}{2\delta_{k}}((u_{k-1}-u_{k})+(a-\delta_{k})(c-a-\delta_{k}))\right]=0,
\end{aligned}
\end{equation}
i.e.,
\begin{equation}
\begin{aligned}
&(u_{k}-u_{k-1})_{T}=\frac{\delta_{k,T}}{\delta_{k}}((u_{k}-u_{k-1})+(\delta_{k}-a)(c-a-\delta_{k}))\\
&\qquad +\frac{1}{2\delta_{k}}((u_{k}-u_{k-1})+(\delta_{k}-a)(c-a-\delta_{k}))^{2}+4u_{k}\delta_{k}\\
&\qquad -c((u_{k}-u_{k-1})+(\delta_{k}-a)(c-a-\delta_{k}))+\delta_{k,T}(2\delta_{k}-c).
\end{aligned}
\label{3.26}
\end{equation}
From the dependent variable transformation (\ref{ddt}), we have
\begin{equation}
\left(\frac{1}{\rho_{k}}\right)_{T}=-\frac{1}{a}\left(\log{\frac{f_{k}}{f_{k-1}}}\right)_{TT}=-\frac{1}{a}(u_{k}-u_{k-1}).\label{3.27}
\end{equation}
Using (\ref{3.24}), (\ref{3.27}) can be rewritten as
\begin{equation}
\frac{d\delta_{k}}{dT}=u_{k-1}-u_{k}.\label{3.28}
\end{equation}
Substituting (\ref{3.28}) into (\ref{3.26}), we obtain 
\begin{equation}
\left\{
\begin{aligned}
&\frac{d}{dT}(u_{k}-u_{k-1})
=-\frac{u_{k}-u_{k-1}}{\delta_{k}}
 \left((u_{k}-u_{k-1})+(\delta_{k}-a)(c-a-\delta_{k})\right)\\
&\qquad +\frac{1}{2\delta_{k}}
 \left((u_{k}-u_{k-1})+(\delta_{k}-a)(c-a-\delta_{k})\right)^{2}
 +4u_{k}\delta_{k}\\
&\qquad -c\left((u_{k}-u_{k-1})+(\delta_{k}-a)(c-a-\delta_{k})\right)
 -(u_{k}-u_{k-1})(2\delta_{k}-c),\\
&\frac{d\delta_{k}}{dT}=-u_{k}+u_{k-1}.
\end{aligned}
\right.
\label{2HS_delta}
\end{equation}

Now, differentiating the integral form of the hodograph transformation (\ref{hod}) with respect to  $T$, we have
\begin{eqnarray}
\frac{\partial x}{\partial T}=\frac{\partial x_{0}(T)}{\partial T}+\frac{\partial}{\partial T}\int_{0}^{X}\frac{1}{\rho(\bar{X},T)}d\bar{X}=\frac{\partial x_{0}(T)}{\partial T}-\int_{0}^{X}\frac{\partial}{\partial \bar{X}}u(\bar{X},T)d\bar{X}
=\frac{\partial x_{0}(T)}{\partial T}-(u(X,T)-u(0,T)).
\label{3.37}
\end{eqnarray}
This calculation employs the conservation law (\ref{conlow}). 
The differential law (\ref{diflaw}) leads to $\displaystyle\frac{\partial x}{\partial T}=-u(X,T)$. 
Therefore, we obtain the evolution equation for $x_{0}$ from (\ref{3.37}), 
\begin{equation}
\frac{\partial x_{0}(T)}{\partial T}=-u(0,T). \label{condition}
\end{equation}
Following \cite{Hori2025}, we refer to (\ref{condition}) as the hodograph consistency condition.

Since $\delta_k=x_k-x_{k-1}$, the second equation of (\ref{2HS_delta}) gives 
\begin{eqnarray}
\frac{dx_{k}}{dT}-\frac{dx_{k-1}}{dT}=-u_{k}+u_{k-1}.
\label{delta_x}
\end{eqnarray}
Hence,
\begin{eqnarray}
x_{k,T}+u_{k}=C(T),
\label{x_evolution_C}
\end{eqnarray}
where $C(T)$ is independent of  $k$. Thus, the evolution equation for the mesh intervals determines the mesh velocities only up to a $k$-independent function. Evaluating (\ref{x_evolution_C}) at $k=0$ and using the consistency condition (\ref{condition}), we obtain $C(T)=0$. Therefore, (\ref{x_evolution_C}) reduces to
\begin{eqnarray}
\frac{dx_{k}}{dT}=-u_{k}.
\label{x_evolution}
\end{eqnarray}
Together with the first equation of (\ref{2HS_delta}), equation (\ref{x_evolution}) yields the semi-discrete 2-HS equation (\ref{semi-2HS}).
\end{proof}

Dividing the first equation of the semi-discrete analogue of the 2-HS equation (\ref{semi-2HS}) by $\delta_{k}$, we obtain
\begin{equation}
\begin{aligned}
&\frac{(u_{k}-u_{k-1})_{T}}{\delta_{k}}=-\frac{u_{k}-u_{k-1}}{\delta_{k}}\left(\frac{u_{k}-u_{k-1}}{\delta_{k}}+\left(1-\rho_{k}\right)(c-a-\delta_{k})\right)\\
&\qquad +\frac{1}{2}\left(\frac{(u_{k}-u_{k-1})^{2}}{\delta_{k}^{2}}+\frac{2(\delta_{k}-a)(c-a-\delta_{k})(u_{k}-u_{k-1})}{\delta_{k}^{2}}+\left(1-2\rho_{k}+\rho_{k}^{2}\right)(c-a-\delta_{k})^{2}\right)\\
&\qquad +4u_{k}-c\left(\frac{u_{k}-u_{k-1}}{\delta_{k}}+\left(1-\rho_{k}\right)(c-a-\delta_{k})\right)-\frac{u_{k}-u_{k-1}}{\delta_{k}}(2\delta_{k}-c).
\end{aligned}
\label{1.72}
\end{equation}
Since
\begin{eqnarray}
\frac{d}{dT}\left(\frac{u_{k}-u_{k-1}}{\delta_{k}}\right)=\frac{(u_{k}-u_{k-1})_{T}}{\delta_{k}}+\frac{(u_{k}-u_{k-1})^{2}}{\delta_{k}^{2}},
\end{eqnarray}
equation (\ref{1.72}) is rewritten as
\begin{equation}
\begin{aligned}
&\frac{d}{dT}\left(\frac{u_{k}-u_{k-1}}{\delta_{k}}\right)=-\frac{u_{k}-u_{k-1}}{\delta_{k}}\left(1-\rho_{k}\right)(c-a-\delta_{k})\\
&\qquad +\frac{1}{2}\left(\frac{(u_{k}-u_{k-1})^{2}}{\delta_{k}^{2}}+\frac{2(\delta_{k}-a)(c-a-\delta_{k})(u_{k}-u_{k-1})}{\delta_{k}^{2}}+\left(1-2\rho_{k}+\rho_{k}^{2}\right)(c-a-\delta_{k})^{2}\right)\\
&\qquad +4u_{k}-c\left(\frac{u_{k}-u_{k-1}}{\delta_{k}}+\left(1-\rho_{k}\right)(c-a-\delta_{k})\right)-\frac{u_{k}-u_{k-1}}{\delta_{k}}(2\delta_{k}-c).
\end{aligned}
\label{3.48}
\end{equation}
Taking the continuous limit $a\rightarrow 0\quad(\delta_{k}\rightarrow 0)$ in (\ref{3.48}), we have
\begin{eqnarray}
(\partial_{t}-u\partial_{x})u_{x}=-cu_{x}(1-\rho)+\frac{1}{2}\left(u_{x}^{2}+2u_{x}(1-\rho)c+c^{2}(1-2\rho+\rho^{2})\right)+4u-c(u_{x}+c(1-\rho))+cu_{x},\label{3.49}
\end{eqnarray}
because 
\begin{eqnarray}
\frac{d}{dT}=\frac{d}{dt}+\frac{dx_{k}}{dT}\frac{d}{dx_{k}}\rightarrow \partial_{t}-u\partial_{x}\quad \mathrm{for}\quad a\rightarrow 0.
\end{eqnarray}
Rearranging (\ref{3.49}), it follows that
\begin{equation}
(\partial_{t}-u\partial_{x})u_{x}=\frac{1}{2}u_{x}^{2}-\frac{c^{2}}{2}+\frac{\rho^{2} c^{2}}{2}+4u,
\end{equation}
i.e.,
\begin{equation}
u_{tx}-uu_{xx}-\frac{u_{x}^{2}}{2}-4u-\frac{c^{2}}{2}\rho^{2}+\frac{c^{2}}{2}=0,
\end{equation}
which is the first equation of (\ref{2HS}). The mesh equation (\ref{3.28}),
together with $\delta_k=a/\rho_k$ in (\ref{3.24}), has the continuous limit
\begin{equation}
\left(\frac{1}{\rho}\right)_{T}=-u_{X},
\qquad \rho_{T}=\rho^{2}u_{X}.
\end{equation}
Using $\partial_T=\partial_t-u\partial_x$ and
$\partial_X=\rho^{-1}\partial_x$, this relation becomes
$\rho_t=(\rho u)_x$, the second equation of (\ref{2HS}). Thus the continuous
limit of the semi-discrete system reproduces the full 2-HS system.

\begin{example}
The reduced Casoratian $\tau$-functions yield the explicit one- and two-soliton solutions of 
the semi-discrete 2-HS equation through the dependent variable transformation 
\begin{equation}
u_{k}=(\log{f_{k}})_{TT},\quad \frac{1}{\rho_{k}}=1-\frac{1}{a}\left(\log \frac{f_{k}}{f_{k-1}}\right)_{T},\label{disdep}
\end{equation}
and the discrete hodograph transformation 
\begin{equation}
x_{k}=x_{0}(T)+\sum_{l=1}^{k}\frac{a}{\rho_{l}}.\label{dishodograph}
\end{equation}

\vspace{1cm}
\noindent
{\bf One-soliton:}
The $\tau$-functions for the one-soliton solution are given below.
\begin{eqnarray}
f_{k}=(1-ap_{1})^{-k}e^{\xi_{1}}+(1-aq_{1})^{-k}e^{\eta_{1}}\propto1+\left(\frac{1-ap_{1}}{1-aq_{1}}\right)^{-k}e^{\xi_{1}-\eta_{1}},
\end{eqnarray}
\begin{eqnarray}
g_{k}=p_{1}(1-ap_{1})^{-k}e^{\xi_{1}}+q_{1}(1-aq_{1})^{-k}e^{\eta_{1}}\propto1
+\frac{p_{1}}{q_{1}}\left(\frac{1-ap_{1}}{1-aq_{1}}\right)^{-k}e^{\xi_{1}-\eta_{1}},
\end{eqnarray}
where 
\begin{equation}
\xi_{1}=\frac{1}{p_{1}}T+\xi_{1,0},\quad 
\eta_{1}=\frac{1}{q_{1}}T+\eta_{1,0}.
\end{equation}
The parameters $p_{1}$ and $q_{1}$ satisfy $\frac{1}{p_{1}}+\frac{1}{q_{1}}=c$.

Figure~\ref{fig3.1} shows the one-soliton solution for $t=0$.
The soliton moves to the right along the $x$-axis as time progresses.


\begin{figure}[h]
\centering
         \includegraphics[keepaspectratio, scale=0.5]{
         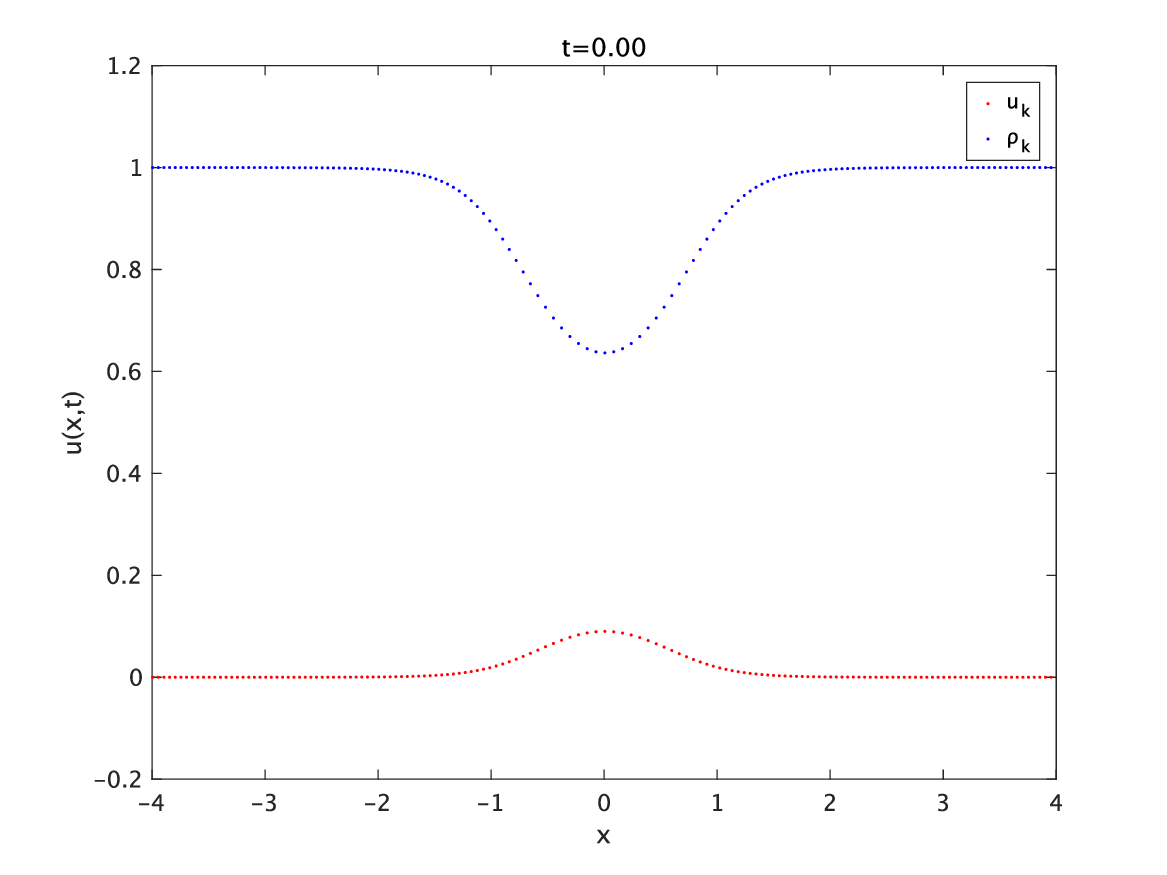}
         \caption{The one-soliton solution of the semi-discrete 2-HS equation (red dots: $u_{k}$, blue dots: $\rho_{k}$). 
         The parameters are $p_{1}=5, c=1, \eta_{1,0}=\xi_{1,0}=0, a=0.005$.}
         \label{fig3.1}

\end{figure}

\vspace{1cm}
\noindent	
{\bf Two-soliton:}
The $\tau$-functions for the two-soliton solution are given below.
As in the continuous case, we write them in the following proportional forms:
\begin{equation*}
\begin{aligned}
&f_{k}=(p_{2}-p_{1})(1-ap_{1})^{-k}(1-ap_{2})^{-k}e^{\xi_{1}+\xi_{2}}+(q_{2}-p_{1})(1-ap_{1})^{-k}(1-aq_{2})^{-k}e^{\xi_{1}+\eta_{2}}\\
&\qquad +(p_{2}-q_{1})(1-ap_{2})^{-k}(1-aq_{1})^{-k}e^{\xi_{2}+\eta_{1}}+(q_{2}-q_{1})(1-aq_{1})^{-k}(1-aq_{2})^{-k}e^{\eta_{1}+\eta_{2}}\\
&\qquad \propto 1+\left(\frac{1-ap_{1}}{1-aq_{1}}\right)^{-k}e^{\xi_{1}-\eta_{1}}+\left(\frac{1-ap_{2}}{1-aq_{2}}\right)^{-k}e^{\xi_{2}-\eta_{2}}\\
&\qquad +\frac{(p_{2}-p_{1})(q_{2}-q_{1})}{(q_{2}-p_{1})(p_{2}-q_{1})}\left(\frac{1-ap_{1}}{1-aq_{1}}\right)^{-k}
\left(\frac{1-ap_{2}}{1-aq_{2}}\right)^{-k}e^{\xi_{1}+\xi_{2}-\eta_{1}-\eta_{2}},
\end{aligned}
\end{equation*}
\begin{equation*}
\begin{aligned}
&g_{k}=p_{1}p_{2}(p_{2}-p_{1})(1-ap_{1})^{-k}(1-ap_{2})^{-k}e^{\xi_{1}+\xi_{2}}+p_{1}q_{2}(q_{2}-p_{1})(1-ap_{1})^{-k}(1-aq_{2})^{-k}e^{\xi_{1}+\eta_{2}}\\
&\qquad +p_{2}q_{1}(p_{2}-q_{1})(1-ap_{2})^{-k}(1-aq_{1})^{-k}e^{\xi_{2}+\eta_{1}}+q_{1}q_{2}(q_{2}-q_{1})(1-aq_{1})^{-k}(1-aq_{2})^{-k}e^{\eta_{1}+\eta_{2}}\\
&\qquad \propto 1+\frac{p_{1}}{q_{1}}\left(\frac{1-ap_{1}}{1-aq_{1}}\right)^{-k}e^{\xi_{1}-\eta_{1}}+\frac{p_{2}}{q_{2}}
\left(\frac{1-ap_{2}}{1-aq_{2}}\right)^{-k}e^{\xi_{2}-\eta_{2}}\\
&\qquad +\frac{p_{1}p_{2}(p_{2}-p_{1})(q_{2}-q_{1})}{q_{1}q_{2}(q_{2}-p_{1})(p_{2}-q_{1})}
\left(\frac{1-ap_{1}}{1-aq_{1}}\right)^{-k}\left(\frac{1-ap_{2}}{1-aq_{2}}\right)^{-k}e^{\xi_{1}+\xi_{2}-\eta_{1}-\eta_{2}},
\end{aligned}
\end{equation*}
where
\begin{equation}
\xi_{i}=\frac{1}{p_{i}}T+\xi_{i,0},\quad
\eta_{i}=\frac{1}{q_{i}}T+\eta_{i,0} \quad \mathrm{for} \quad i=1,2.
\end{equation}
The parameters $p_{i}$ and $q_{i}$ satisfy $\frac{1}{p_{i}}+\frac{1}{q_{i}}=c$ for $i=1,2$.

Figure~\ref{fig3.2} shows the two-soliton interaction of $u_{k}$
and $\rho_{k}$. Both solitons move to the right along the $x$-axis as time
progresses.


\begin{figure}[h]
  \begin{tabular}{cc}
      \begin{minipage}[t]{0.45\hsize}
              \includegraphics[width=\linewidth]{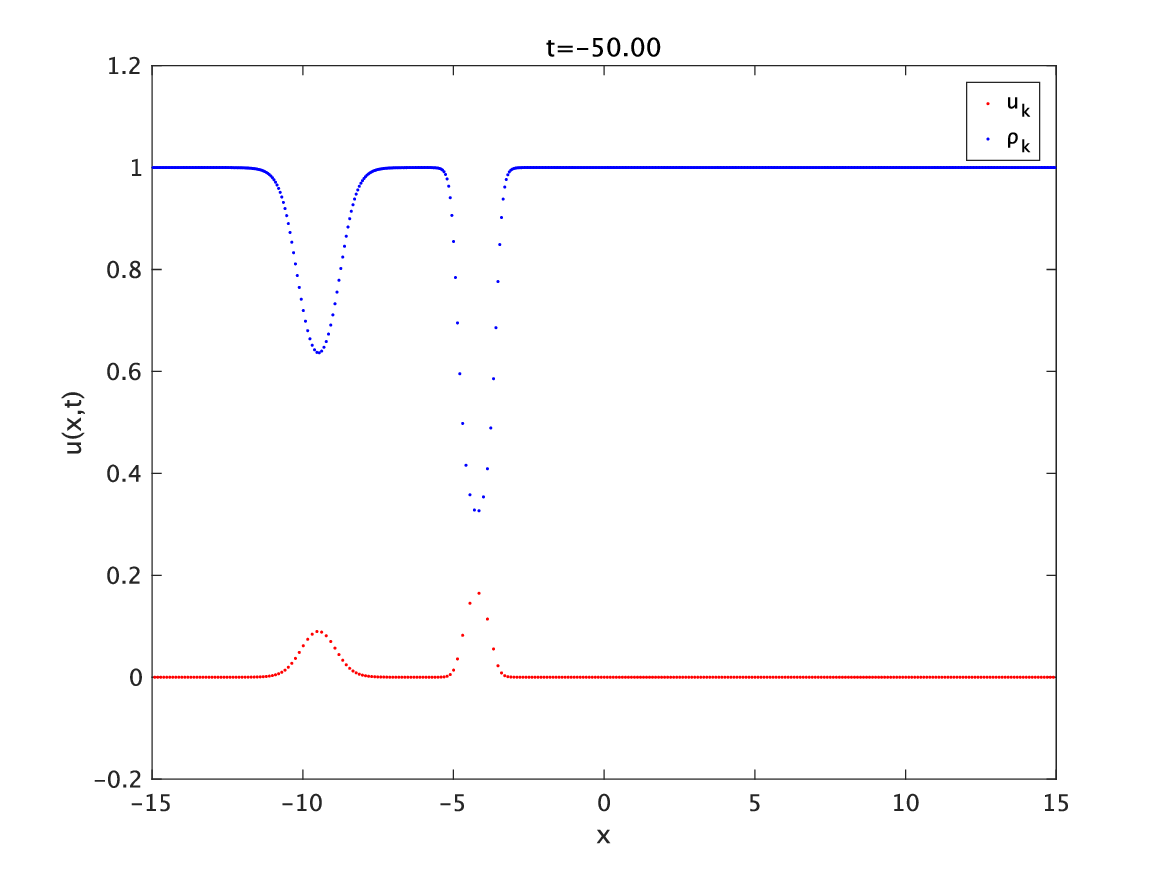}
        \captionsetup{labelformat=empty,labelsep=none}
      \end{minipage}
                \begin{minipage}[t]{0.45\hsize}
        \includegraphics[width=\linewidth]{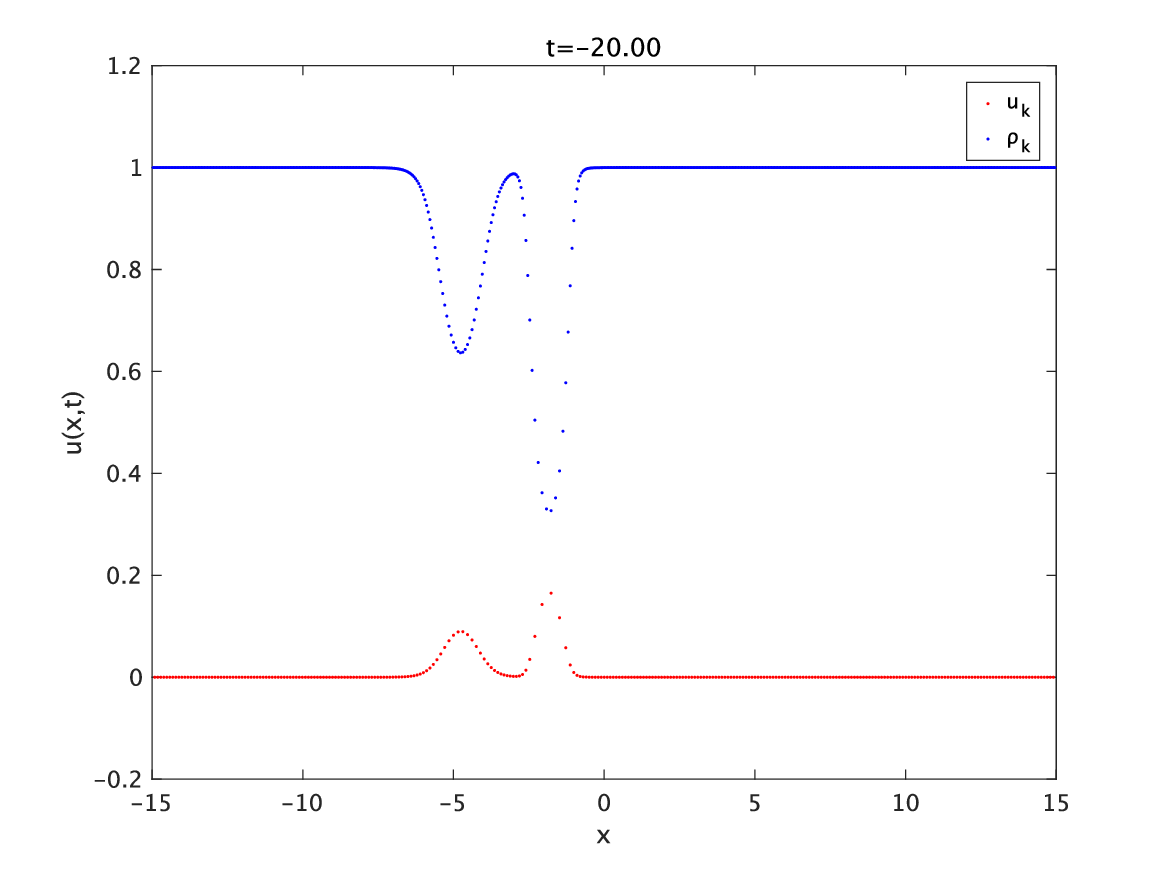}
        \captionsetup{labelformat=empty,labelsep=none}
      \end{minipage}
                     \end{tabular}
\\

 \begin{tabular}{cc}
      \begin{minipage}[t]{0.45\hsize}
        \includegraphics[width=\linewidth]{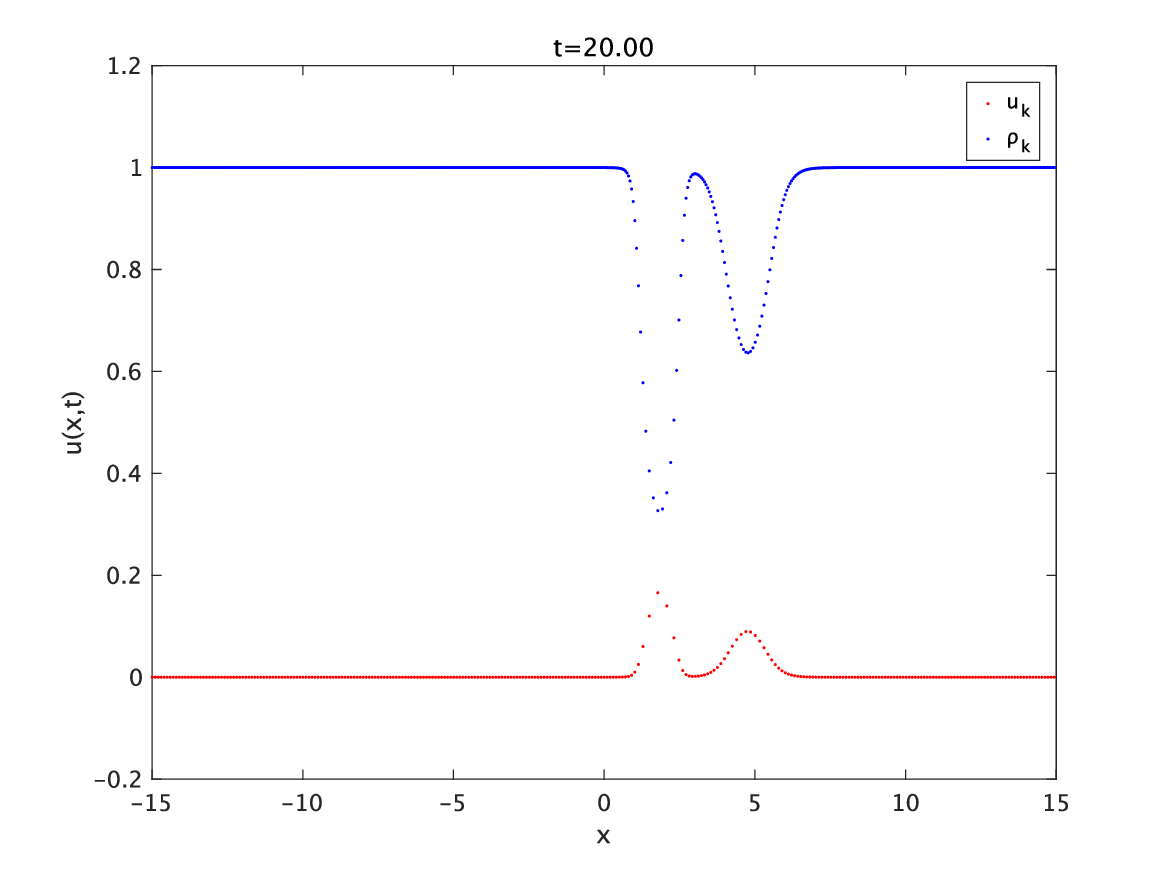}
        \captionsetup{labelformat=empty,labelsep=none}
      \end{minipage}

      \begin{minipage}[t]{0.45\hsize}
        \includegraphics[width=\linewidth]{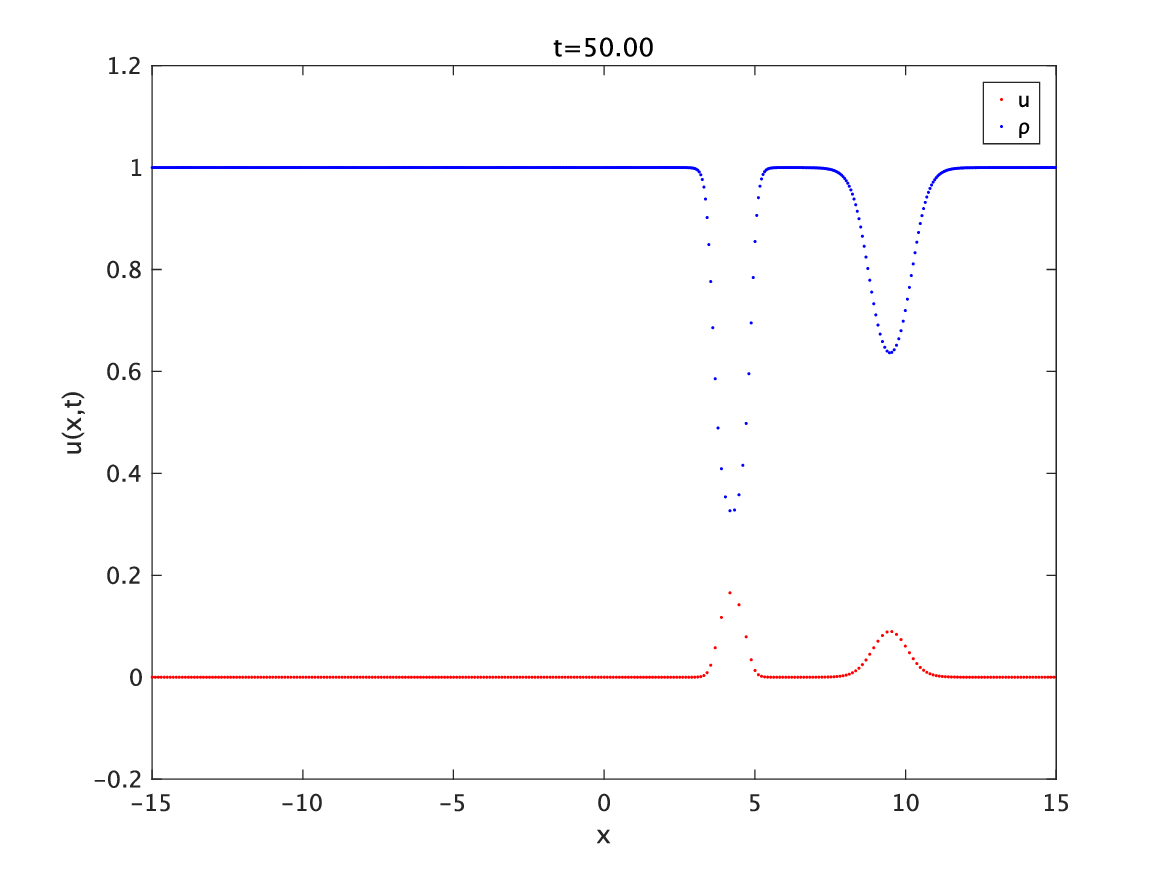}
        \captionsetup{labelformat=empty,labelsep=none}
      \end{minipage} 
         \end{tabular}
         \caption{The two-soliton interaction of the semi-discrete 2-HS equation (red dots: $u_{k}$, blue dots: $\rho_{k}$). 
         The parameters are 
         $p_{1}=1.1, p_{2}=1.25, c=1, a=0.005, 
         \xi_{1,0}=\frac{1}{2}\log{((q_{1} - p_{2}) (q_{1}- q_{2}))}, 
         \eta_{1,0}= \frac{1}{2} \log{((p_{2} - p_{1}) (q_{2} - p_{1}))} + \pi \mathrm{i}, \xi_{2,0}=\frac{1}{2}\log{((q_{2} - p_{1}) (q_{1} - q_{2}))}, 
         \eta_{2,0}=\frac{1}{2}\log{((p_{2} - p_{1}) (q_{1} - p_{2}))}$.}
         \label{fig3.2}
         
  \end{figure}

\end{example}

\section{Lax representation and continuous limit}
\label{sec:lax}

In this section, we first present a $2\times2$ matrix Lax pair for the semi-discrete 2-HS equation and derive its scalar reduction. We then take the continuous limit of the matrix Lax pair and derive the corresponding continuous scalar Lax pair. Finally, we show that the resulting scalar Lax pair agrees with the known representation of Aratyn et al. \cite{AGZ}.

\subsection{Matrix Lax pair}

\begin{proposition}[Matrix Lax pair]
Let
\begin{equation}
\alpha:=a(a-c),\qquad
\Delta(\lambda):=1-\alpha\lambda.
\end{equation}
Assume that $\lambda\neq 0$, $\Delta(\lambda)\neq 0$, and $\delta_{k}\neq 0$ for all $k$.
Then the semi-discrete 2-HS equation (\ref{semi-2HS}) admits the Lax pair
\begin{equation}
\Psi_{k}=U_{k}(\lambda)\Psi_{k-1},\qquad
\frac{\partial \Psi_{k}}{\partial T}=V_{k}(\lambda)\Psi_{k},
\label{sd-lax}
\end{equation}
where $\lambda$ is a $T$-independent spectral parameter, 
$\Psi_{k}=\Psi_{k}(T,\lambda)$ is a two-component vector, and
\begin{equation}
V_{k}(\lambda)=
\left(
\begin{array}{cc}
0 & 1 \\
\displaystyle -2u_{k}+\frac{c^{2}}{4}+\frac{1}{2\lambda} & 0
\end{array}
\right),
\label{sd-lax-V}
\end{equation}
\begin{equation}
U_{k}(\lambda)=\frac{1}{\Delta(\lambda)}
\left(
\begin{array}{cc}
\displaystyle 1+\lambda\left(\delta_{k}^{2}+u_{k}-u_{k-1}-\alpha\right) & \displaystyle 2\lambda\delta_{k} \\
\displaystyle \delta_{k}+\frac{\lambda}{2\delta_{k}}\left((\alpha-\delta_{k}^{2})^{2}-(u_{k}-u_{k-1})^{2}\right) & \displaystyle 1+\lambda\left(\delta_{k}^{2}-u_{k}+u_{k-1}-\alpha\right)
\end{array}
\right).
\label{sd-lax-U}
\end{equation}
The compatibility condition
\begin{equation}
\frac{\partial U_{k}}{\partial T}+U_{k}V_{k-1}-V_{k}U_{k}=0
\label{sd-lax-comp}
\end{equation}
is equivalent to the first equation of (\ref{semi-2HS}) together with
$\delta_{k,T}=u_{k-1}-u_{k}$. Together with
$\delta_k=x_k-x_{k-1}$ and the hodograph consistency condition
$x_{0,T}=-u_0$, these equations are equivalent to the full semi-discrete
2-HS system (\ref{semi-2HS}).
\end{proposition}

\begin{proof}
Substituting (\ref{sd-lax-V}) and (\ref{sd-lax-U}) into (\ref{sd-lax-comp}) and clearing the common denominator $\Delta(\lambda)$, the $(1,2)$-entry gives 
\begin{equation} \delta_{k,T}=u_{k-1}-u_k. 
\label{sd-lax-delta} 
\end{equation} 
Using this relation, the $(1,1)$-entry reduces exactly to the first equation of (\ref{semi-2HS}), while the $(2,1)$- and $(2,2)$-entries give no additional conditions. Thus the zero-curvature condition is equivalent to the first equation of (\ref{semi-2HS}) together with (\ref{sd-lax-delta}).  As shown in (\ref{delta_x})--(\ref{x_evolution}), equation (\ref{sd-lax-delta}), together with $\delta_k=x_{k}-x_{k-1}$ and the hodograph consistency condition $x_{0,T}=-u_{0}$, yields $x_{k,T}=-u_k$. Hence, under the hodograph consistency condition, (\ref{sd-lax-comp}) is equivalent to the semi-discrete 2-HS system (\ref{semi-2HS}).
\end{proof}

Since $\Delta(\lambda)$ is independent of $k$ and $T$, it can be absorbed by a scalar gauge transformation and therefore does not affect the zero-curvature condition.

\subsection{Scalar Lax pair}

Eliminating the second component of the matrix eigenfunction gives a scalar
Lax pair. Write
\begin{equation}
\Psi_{k}=\begin{pmatrix}\phi_{k}\\ \chi_{k}\end{pmatrix},
\end{equation}
and introduce
\begin{equation}
\begin{aligned}
\mathcal A_{k}&=1+\lambda\left(\delta_{k}^{2}+u_{k}-u_{k-1}-\alpha\right),
&\mathcal B_{k}&=2\lambda\delta_{k},\\
\mathcal C_{k}&=\delta_{k}+\frac{\lambda}{2\delta_{k}}
\left\{\left(\alpha-\delta_{k}^{2}\right)^{2}-\left(u_{k}-u_{k-1}\right)^{2}\right\},
&\mathcal E_{k}&=1+\lambda\left(\delta_{k}^{2}-u_{k}+u_{k-1}-\alpha\right).
\end{aligned}
\label{sd-scalar-entries}
\end{equation}
Thus
$U_{k}=\Delta(\lambda)^{-1}
\left(\begin{smallmatrix}\mathcal A_{k}&\mathcal B_{k}\\
\mathcal C_{k}&\mathcal E_{k}\end{smallmatrix}\right)$. 

The first row of the time equation gives $\chi_{k}=\phi_{k,T}$. On the
other hand, the first row of the spatial equation at the lattice point
$k+1$ gives
\begin{equation}
\chi_{k}=\frac{\Delta(\lambda)\phi_{k+1}-\mathcal A_{k+1}\phi_{k}}
{\mathcal B_{k+1}}.
\label{chi_k}
\end{equation}
Equating these two expressions for $\chi_k$ gives the scalar time evolution
\begin{equation}
\phi_{k,T}=\frac{\Delta(\lambda)\phi_{k+1}-\mathcal A_{k+1}\phi_{k}}
{\mathcal B_{k+1}}.
\label{sd-scalar-time}
\end{equation}
To obtain the scalar spatial problem, substitute (\ref{chi_k}) and its $k\rightarrow k-1$ shift into the second row of the spatial equation $\Delta(\lambda)\chi_{k}=\mathcal C_{k}\phi_{k-1}
+\mathcal E_{k}\chi_{k-1}$. Using
\begin{equation}
\mathcal A_{k}\mathcal E_{k}-\mathcal B_{k}\mathcal C_{k}
=P(\lambda),\qquad P(\lambda):=1-2\alpha\lambda,
\label{sd-scalar-det}
\end{equation}
we obtain the scalar spatial equation.
\begin{equation}
\begin{aligned}
\frac{\Delta(\lambda)^{2}}{\mathcal B_{k+1}}\phi_{k+1}
&-\Delta(\lambda)\left(
\frac{\mathcal A_{k+1}}{\mathcal B_{k+1}}
+\frac{\mathcal E_{k}}{\mathcal B_{k}}
\right)\phi_{k}
+\frac{P(\lambda)}{\mathcal B_{k}}\phi_{k-1}=0.
\end{aligned}
\label{sd-scalar-space}
\end{equation}
Under the assumptions of Proposition 1, (\ref{sd-scalar-time}) and (\ref{sd-scalar-space}) are equivalent to the scalar reduction of the matrix Lax pair. Their compatibility follows from the matrix zero-curvature condition.

\subsection{Continuous limit and comparison with the Aratyn-Gomes-Zimerman scalar Lax pair}

We now take the continuous limit $a\to0$ of the matrix Lax pair.
Let $X_{k}=ka$ and identify $\Psi_{k}(T,\lambda)$ with a smooth
vector $\Psi(X,T,\lambda)$. Using
\begin{equation}
\delta_{k}=\frac{a}{\rho_{k}},\qquad
u_{k}-u_{k-1}=a u_{X}+O(a^{2}),\qquad
\Delta(\lambda)=1+ac\lambda+O(a^{2}),
\label{sd-lax-cont-expansions}
\end{equation}
we obtain
\begin{equation}
U_{k}(\lambda)=I+a{\tilde{L}}(\lambda)+O(a^{2}),
\qquad
{\tilde{L}}(\lambda)=
\begin{pmatrix}
\lambda u_{X} & \displaystyle\frac{2\lambda}{\rho} \\
\displaystyle\frac{1}{\rho}+\frac{\lambda\rho}{2}
\left(c^{2}-u_{X}^{2}\right) & -\lambda u_{X}
\end{pmatrix}.
\label{sd-lax-cont-X}
\end{equation}
Hence the continuous limit of (\ref{sd-lax}) is
\begin{equation}
\Psi_{X}={\tilde{L}}(\lambda)\Psi,\qquad
\Psi_{T}=V(\lambda)\Psi,\qquad
V(\lambda)=
\begin{pmatrix}
0 & 1 \\
\displaystyle -2u+\frac{c^{2}}{4}+\frac{1}{2\lambda} & 0
\end{pmatrix}.
\label{sd-lax-cont-hodo}
\end{equation}
Using the differential law (\ref{diflaw}), writing
$\hat{\Psi}(x,t,\lambda)=\Psi(X,T,\lambda)$ gives the Lax pair in the physical variables,
\begin{equation}
\hat{\Psi}_{x}=L(\lambda)\hat{\Psi},\qquad
(\partial_{t}-u\partial_{x})\hat{\Psi}=V(\lambda)\hat{\Psi},
\label{sd-lax-cont-physical}
\end{equation}
where
\begin{equation}
L(\lambda)=
\begin{pmatrix}
\lambda u_{x} & 2\lambda \\
\displaystyle 1+\frac{\lambda}{2}
\left(c^{2}\rho^{2}-u_{x}^{2}\right) & -\lambda u_{x}
\end{pmatrix}.
\label{sd-lax-cont-L}
\end{equation}
To obtain the scalar Lax pair corresponding to the continuous limit, write the two-component eigenfunction as
$\hat{\Psi}=(\phi,\chi)^{\mathsf T}$. The spatial equations are
\begin{equation}
\begin{aligned}
\phi_{x}&=\lambda u_{x}\phi+2\lambda\chi,\\
\chi_{x}&=\left\{1+\frac{\lambda}{2}
\left(c^{2}\rho^{2}-u_{x}^{2}\right)\right\}\phi
-\lambda u_{x}\chi.
\end{aligned}
\label{sd-lax-cont-components}
\end{equation}
The first equation gives
\begin{equation}
\chi=\frac{1}{2\lambda}\phi_{x}-\frac{1}{2}u_{x}\phi,
\qquad
\chi_{x}=\frac{1}{2\lambda}\phi_{xx}
-\frac{1}{2}u_{xx}\phi-\frac{1}{2}u_{x}\phi_{x}.
\label{sd-lax-cont-eliminate}
\end{equation}
On the other hand, substituting the first relation in
(\ref{sd-lax-cont-eliminate}) into the second equation of
(\ref{sd-lax-cont-components}) yields
\begin{equation}
\chi_{x}=\left(1+\frac{\lambda c^{2}\rho^{2}}{2}\right)\phi
-\frac{1}{2}u_{x}\phi_{x}.
\end{equation}
Comparing these two expressions for $\chi_{x}$ gives the scalar spatial
spectral problem. Moreover, the first row of the time equation in
(\ref{sd-lax-cont-physical}) is $\phi_{t}-u\phi_{x}=\chi$.
Substitution of (\ref{sd-lax-cont-eliminate}) gives the scalar time evolution.
Thus, we obtain the scalar Lax pair
\begin{equation}
\phi_{xx}=\left\{\lambda(u_{xx}+2)+c^{2}\lambda^{2}\rho^{2}\right\}\phi,
\qquad
\phi_{t}=\left(u+\frac{1}{2\lambda}\right)\phi_{x}
-\frac{1}{2}u_{x}\phi.
\label{sd-lax-cont-scalar}
\end{equation}

A scalar Lax representation for the continuous 2-HS equation was given by Aratyn, Gomes, and Zimerman (AGZ) \cite{AGZ} in their unified deformation framework for the two-component CH and Dym/HS-type systems. To avoid confusion with the coefficients $\mu$
and $\kappa$ in (\ref{Feng}) and (\ref{2CH}), we attach the subscript
$\mathrm A$ to the parameters and fields in the AGZ formulation. Their Lax pair reads
\begin{equation}
\Phi_{xx}=\left(\frac{\mu_{\mathrm A}^{2}}{4}-M_{\mathrm A}\Lambda_{\mathrm A}
+R_{\mathrm A}^{2}\Lambda_{\mathrm A}^{2}\right)\Phi,
\qquad
\Phi_{t}=-\left(\frac{1}{2\Lambda_{\mathrm A}}+U_{\mathrm A}\right)\Phi_{x}
+\frac{1}{2}(U_{\mathrm A})_{x}\Phi,
\label{agz-general-scalar-pair}
\end{equation}
where
\begin{equation}
M_{\mathrm A}=\mu_{\mathrm A}^{2}U_{\mathrm A}-(U_{\mathrm A})_{xx}
+\frac{\kappa_{\mathrm A}}{2}.
\end{equation}
Here, $\kappa_{\mathrm A}/2$ is an integration constant. In the Dym/HS reduction
$\mu_{\mathrm A}=0$, this term remains in the spatial spectral problem and accounts for
the nonzero linear term in the present 2-HS equation. Indeed, setting $\mu_{\mathrm A}=0$ and making the identification
\begin{equation}
U_{\mathrm A}=-u,\qquad \Lambda_{\mathrm A}=-\lambda,\qquad
R_{\mathrm A}=c\rho,\qquad \kappa_{\mathrm A}=4,
\end{equation}
the scalar Lax pair of Aratyn et al. reduces exactly to (\ref{sd-lax-cont-scalar}). Hence the continuous limit of the present semi-discrete Lax pair reproduces the known scalar Lax representation of the 2-HS equation.

\section{Conclusion}
\label{4}
In this paper, we have constructed an integrable semi-discrete analogue of the 2-HS equation, the short-wave limit of the 2-CH equation. In Section \ref{2}, we introduced a new set of bilinear equations, distinct from the conventional CH-type bilinear form of 
Lou, Feng, and Yao \cite{2HS-10}, and showed that, through a pseudo 2-reduction and a hodograph transformation, these bilinear equations lead to the 2-HS equation. We also constructed the N-soliton solutions in Wronskian form, and examples of one- and two-soliton solutions were presented. 
In Section \ref{3}, we discretized this new bilinear formulation 
in the spatial direction, constructed the N-soliton solutions in Casoratian form, 
and presented the semi-discrete one- and two-soliton solutions. By applying 
the discrete hodograph transformation and the dependent variable transformation, we derived 
an integrable semi-discrete analogue of the 2-HS equation. In Section \ref{sec:lax},
we constructed a matrix Lax pair for the semi-discrete system, derived the corresponding scalar Lax pair, and showed that the continuous limit of the matrix Lax pair yields, after scalar reduction, the known scalar Lax representation of the continuous 2-HS equation.
To the best of our knowledge, this is the first integrable semi-discretization 
that preserves the two-component structure of the 2-HS equation. 
In the Sato-theoretic sense, its integrability is inherited from the underlying hierarchy 
through the bilinear tau-function construction, the pseudo 2-reduction, 
and the discrete hodograph transformation. 
The Lax pair complements this construction by providing an additional direct confirmation 
of the integrable structure in the physical variables.

Future work includes constructing a fully discrete integrable 
analogue of the 2-HS equation, as well as integrable spatial discretizations 
of related systems, such as the 2-CH equation.

\appendix
\section{Cuspon solutions of the Hunter--Saxton equation}
\label{app:cuspon}
\setcounter{equation}{0}
\setcounter{figure}{0}
\renewcommand{\theequation}{A.\arabic{equation}}
\renewcommand{\thefigure}{A.\arabic{figure}}

Taking the limit $c\rightarrow0$ in the bilinear equations (\ref{2.17}), 
with the parameter constraint (\ref{2ps}) reducing to the 2-reduction $q_{i}=-p_{i}$, we obtain the bilinear equations of the HS equation
\begin{eqnarray}
\left\{
\begin{array}{ll}
D_{T}^{2}g\cdot f=0,\\
 (D_{X}D_{T}^{2}-4D_{T})g\cdot f=0.
 \end{array}
\right.
\label{A1}
\end{eqnarray}
The bilinear equations (\ref{A1}) admit the N-soliton solution, expressed as 
\begin{equation}
f=\tau_{0}, \quad g=\tau_{1},\qquad
\tau_{n}=
\left|
\begin{array}{rrrr}
\varphi_{1}^{(n)}& \varphi_{1}^{(n+1)}& \cdots & \varphi_{1}^{(n+N-1)} \\
\varphi_{2}^{(n)}& \varphi_{2}^{(n+1)}& \cdots & \varphi_{2}^{(n+N-1)} \\
\vdots \quad& \vdots \quad & \ddots & \vdots \quad\quad\\
\varphi_{N}^{(n)}& \varphi_{N}^{(n+1)}& \cdots & \varphi_{N}^{(n+N-1)} \\
\end{array}
\right|,\label{A3}
\end{equation}
where 
\begin{eqnarray}
\varphi_{i}^{(n)}=p_{i}^{n}e^{p_{i}X+\frac{1}{p_{i}}T+\xi_{i,0}}+(-p_{i})^{n}e^{-p_{i}X-\frac{1}{p_{i}}T+\eta_{i,0}}.\label{A4}
\end{eqnarray}
The N-cuspon solution of the HS equation is given through the following dependent variable and hodograph transformations:
\begin{eqnarray}
u=(\log{f})_{TT},\quad x=X-(\log{f})_{T},\quad t=T,\quad \rho=\frac{1}{1-(\log f)_{XT}}.\label{A5}
\end{eqnarray}
Following the same procedure as in Theorem~\ref{thm:continuous-2hs}, the bilinear equations (\ref{A1}) lead to
\begin{eqnarray}
(\rho u_{X})_{T}-\frac{1}{2}(\rho u_{X})^{2}-4u=0.\qquad
\label{A55}
\end{eqnarray}
Applying the differential law (\ref{diflaw}), (\ref{A55}) is rewritten as 
\begin{eqnarray}
u_{tx}-uu_{xx}-\frac{1}{2}u_{x}^{2}-4u=0.
\end{eqnarray}
Differentiating this equation with respect to $x$, we obtain
\begin{eqnarray}
u_{txx}-4u_{x}-2u_{x}u_{xx}-uu_{xxx}=0,
\end{eqnarray}
which is exactly the HS equation (\ref{HS}) under the change $u \rightarrow -u$ and the choice $\kappa^{2}=2$.

The one- and two-cuspon solutions are obtained by taking the limit $c\rightarrow0$ in the 2-HS soliton formulas, with the constraint (\ref{2ps}) reducing to $q_{i}=-p_{i}.$
In what follows, we present examples of the one- and two-cuspon solutions. For $N=1$, the one-cuspon solution is
\begin{equation}
\begin{aligned}
u&=\frac{1}{p_{1}^{2}}\mathrm{sech}^{2}{\left(p_{1}X+\frac{1}{p_{1}}T+\frac{\xi_{1,0}-\eta_{1,0}}{2}\right)},\\
x&=X-\frac{1}{p_{1}}\tanh\left(p_{1}X+\frac{1}{p_{1}}T+\frac{\xi_{1,0}-\eta_{1,0}}{2}\right),
\qquad t=T.
\end{aligned}
\end{equation}

\begin{figure}[h]
\centering
         \includegraphics[keepaspectratio, scale=0.5]{
         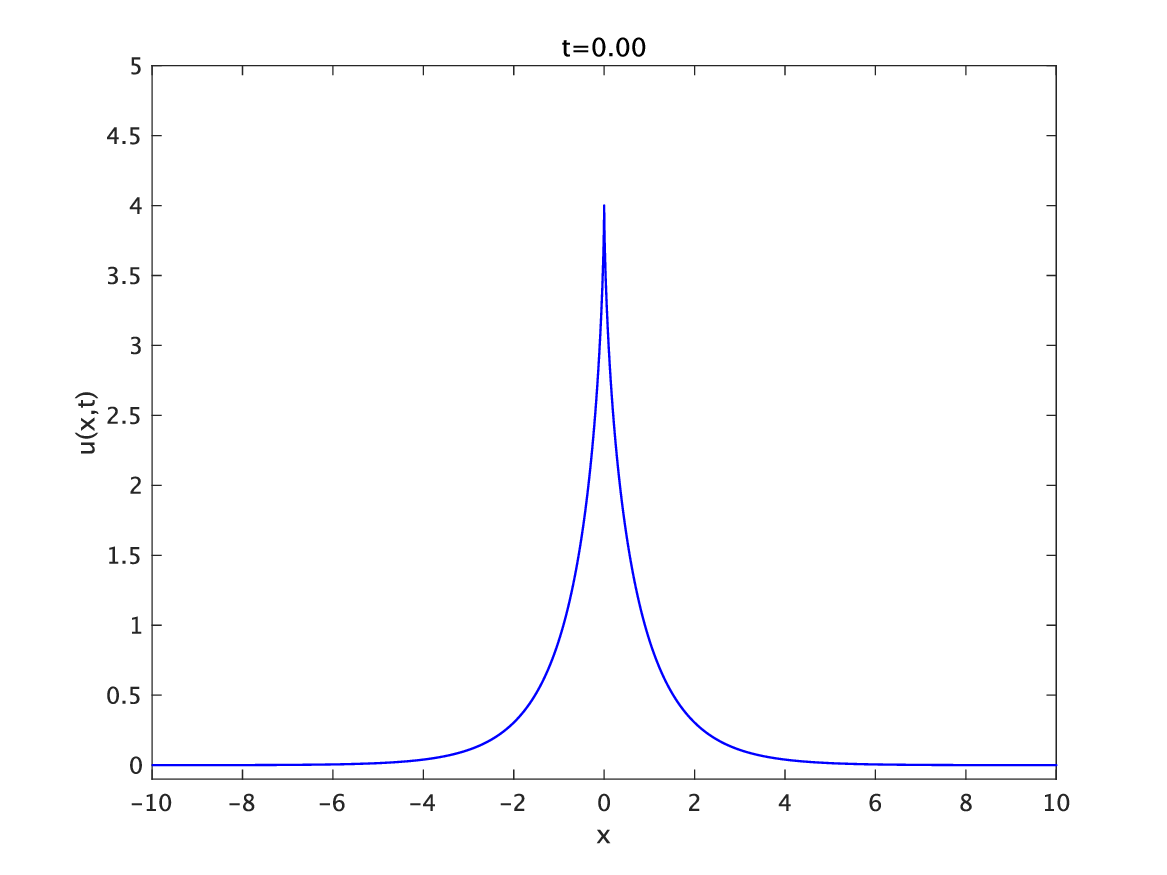}
         \caption{One-cuspon solution of the HS equation. 
         The parameters are $p_{1}=\frac{1}{2}, \xi_{1,0}=\eta_{1,0}=0$.} 
         \label{figA.1}

\end{figure}

\begin{figure}[h]
  \begin{tabular}{cc}
      \begin{minipage}[t]{0.45\hsize}
              \includegraphics[width=\linewidth]{
              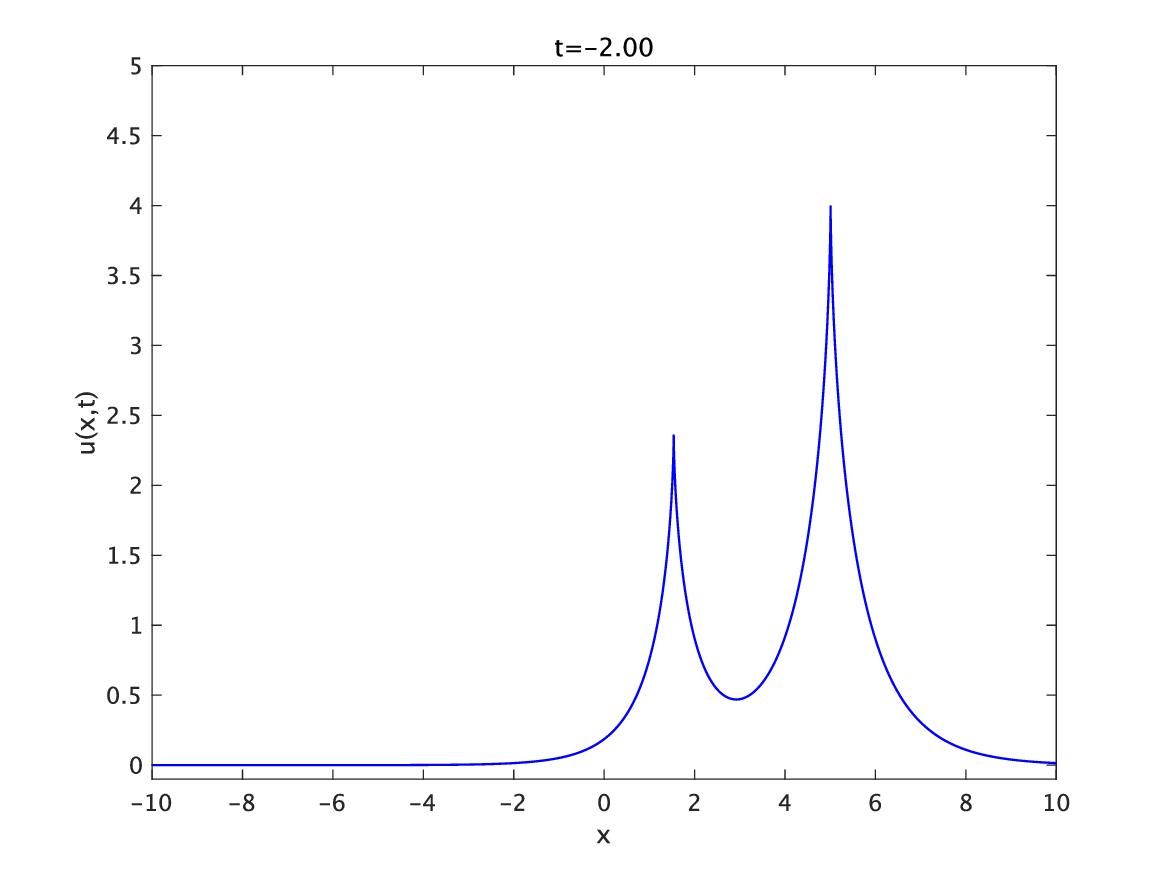}
        \captionsetup{labelformat=empty,labelsep=none}
      \end{minipage}
                \begin{minipage}[t]{0.45\hsize}
        \includegraphics[width=\linewidth]{
        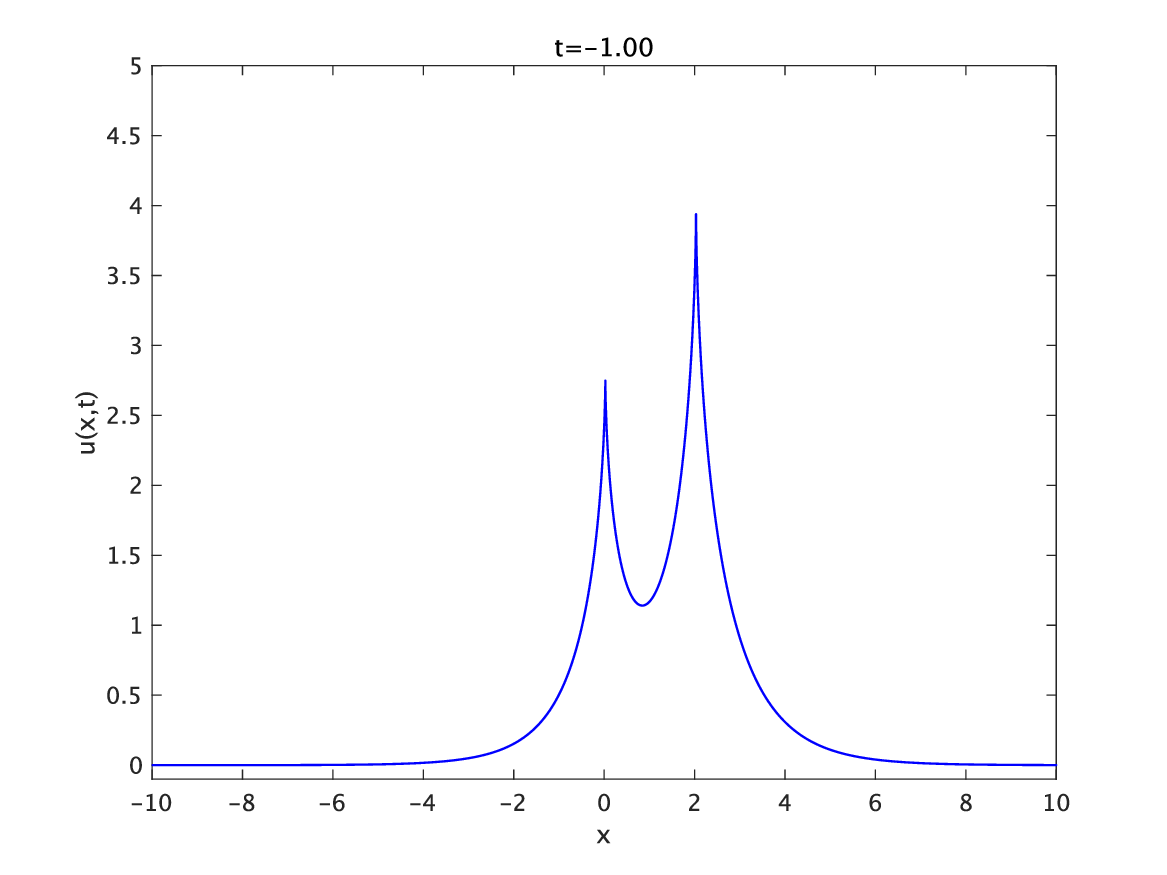}
        \captionsetup{labelformat=empty,labelsep=none}
      \end{minipage}
                     \end{tabular}
\\

 \begin{tabular}{cc}
      \begin{minipage}[t]{0.45\hsize}
        \includegraphics[width=\linewidth]{
        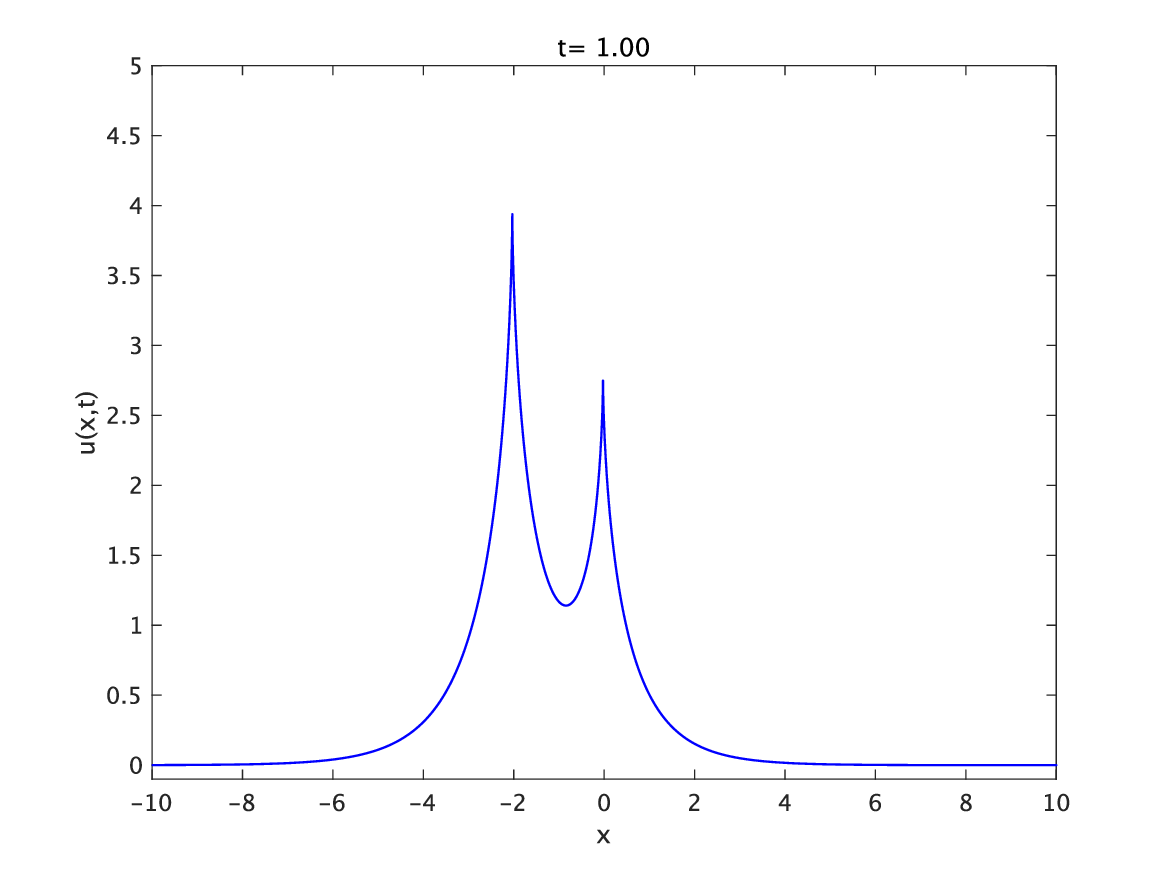}
        \captionsetup{labelformat=empty,labelsep=none}
      \end{minipage}

      \begin{minipage}[t]{0.45\hsize}
        \includegraphics[width=\linewidth]{
        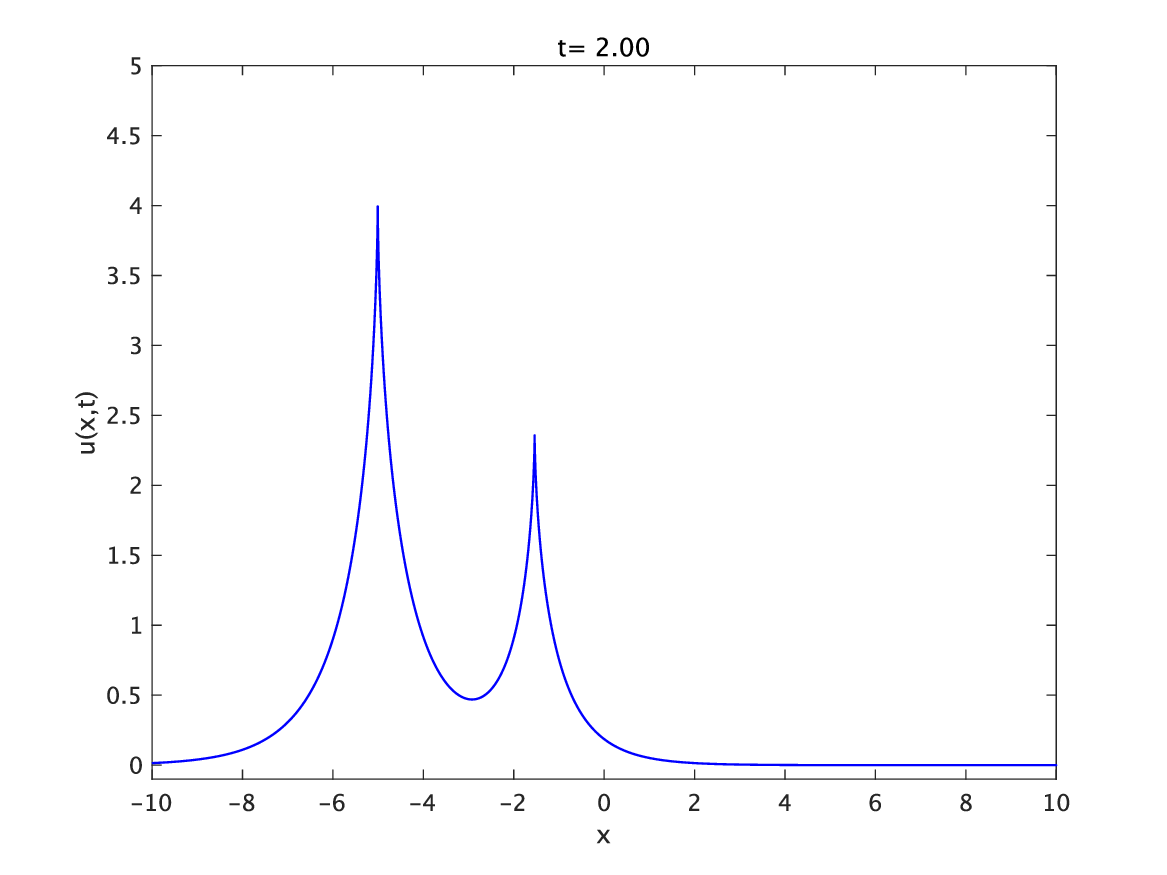}
        \captionsetup{labelformat=empty,labelsep=none}
      \end{minipage} 
         \end{tabular}
         \caption{Two-cuspon interaction of the HS equation. 
         The parameters are $p_{1}=\frac{1}{2}, p_{2}=\frac{2}{3}, \xi_{1,0}=\eta_{1,0}=\eta_{2,0}=0, 
         \xi_{2,0}=\pi\mathrm{i}$.}
         \label{figA.2}
         
  \end{figure}

Figure~\ref{figA.1} shows the one-cuspon solution. Its amplitude
is $1/p_{1}^{2}$, and its velocity in the $(X,T)$ plane is $-1/p_{1}^{2}$.
Figure~\ref{figA.2} shows the two-cuspon interaction.

The one- and two-cuspon solutions propagate to the left along the $x$-axis as time progresses.
Moreover, the N-cuspon solution (\ref{A3})--(\ref{A5}) is consistent with
the N-cuspon solution of the HS equation presented in the previous study \cite{maruno2}.

\section*{CRediT authorship contribution statement}
\textbf{Ayako Hori:} Conceptualization, Methodology, Formal analysis, Investigation, Validation, Visualization, Writing--original draft, Writing--review \& editing.
\textbf{Yuta Tanaka:} Conceptualization, Methodology, Formal analysis, Validation, Writing--review \& editing.
\textbf{Ken-ichi Maruno:} Conceptualization, Methodology, Formal analysis, Validation, Supervision, Writing--review \& editing.
\textbf{Yasuhiro Ohta:} Conceptualization, Methodology, Formal analysis, Validation, Writing--review \& editing.

\section*{Declaration of competing interest}
The authors declare that they have no known competing financial interests or personal relationships 
that could have appeared to influence the work reported in this paper.

\section*{Data availability}
No data were used for the research described in this article.

\section*{Declaration of generative AI and AI-assisted technologies in the manuscript preparation process}
During the preparation of this work, the authors used OpenAI's GPT-5.6 (via Codex) to assist with language editing, the identification of potentially
relevant prior literature and the checking of historical context, and the
development of algebraic and symbolic computation programs for REDUCE and
Mathematica used to verify the Lax pair calculations. The authors
independently consulted the original publications and verified the citations
and historical statements. After using this tool, the authors reviewed and
edited the content as needed and take full responsibility for the content of
the publication.

\section*{Acknowledgements}
This work was partially supported by JSPS KAKENHI Grant Numbers JP22K03441, JP23K22407, and JP26K06919, 
and by Waseda University Grants for Special Research Projects.

\end{document}